\documentclass{sig-alternate-2013}
\pdfoutput=1

\usepackage{etoolbox}

\usepackage[usenames,dvipsnames]{color}

\usepackage{algorithm}
\usepackage{algorithmic}
\usepackage{graphicx} 
\usepackage{amsmath}
 \usepackage{amssymb}
\usepackage[caption=false]{subfig}

\usepackage{flushend}

\usepackage{tikz}
\usetikzlibrary{shapes,arrows,shadows,positioning,matrix,backgrounds}

\newtheorem{theorem}{Theorem}
\newtheorem{prop}{Proposition}

\newtheorem{lem}{Lemma}

\newcommand{\vecn}{{\bf n}}

\newcommand{\vecalpha}{{\boldsymbol \alpha}}

\newcommand{\pnorm}[2]{{\Vert #1 \Vert} _{#2}}

\newcommand{\E}[1]{\mathbb{E}\left[#1 \right]}

\newcommand{\INDSTATE}[1][1]{\STATE\hspace{#1em}} 

\newcommand{\trian}{\textsc{trian}}
\newcommand{\prof}{\textsc{3-prof}}
\newcommand{\egopar}{\textsc{Ego-par}}
\newcommand{\egoser}{\textsc{Ego-ser}}

\usepackage{hyperref}

\begin{document}

\setcounter{MaxMatrixCols}{20} 


%
\CopyrightYear{2015} 
\title{Beyond Triangles: A Distributed Framework for Estimating 3-profiles of Large Graphs
\titlenote{The authors would like to acknowledge support from NSF CCF-1344364, NSF CCF-1344179, ARO YIP W911NF-14-1-0258, DARPA XDATA, and research gifts by Google and Docomo. \\
This work will be presented in part at KDD'15.
}
}

\toappear{}

%
%
%
%

\numberofauthors{4} 
%
\author{
%
%
\alignauthor
Ethan R. Elenberg\\
       \affaddr{The University of Texas}\\
       \affaddr{Austin, Texas 78712, USA}\\
       \email{elenberg@utexas.edu}
\alignauthor
Karthikeyan Shanmugam\\
       \affaddr{The University of Texas}\\
       \affaddr{Austin, Texas 78712, USA}\\
       \email{karthiksh@utexas.edu}
\and  
\alignauthor Michael Borokhovich\\
       \affaddr{The University of Texas}\\
       \affaddr{Austin, Texas 78712, USA}\\
       \email{michaelbor@utexas.edu}
\alignauthor Alexandros G. Dimakis\\
       \affaddr{The University of Texas}\\
       \affaddr{Austin, Texas 78712, USA}\\
       \email{dimakis@austin.utexas.edu}
}
\maketitle

\begin{abstract}
We study the problem of approximating the $3$-profile of a large graph. $3$-profiles are generalizations of triangle counts that specify the number of times a small graph appears as an induced subgraph of a large graph. Our algorithm uses the novel concept of $3$-profile sparsifiers: sparse graphs that can be used to approximate the full $3$-profile counts for a given large graph. Further, we study the problem of estimating local and ego $3$-profiles, two graph quantities that characterize the local neighborhood of each vertex of a graph. 

Our algorithm is distributed and operates as a vertex program over the GraphLab PowerGraph framework. We introduce the concept of edge pivoting which allows us to collect $2$-hop information without maintaining an explicit $2$-hop neighborhood list at each vertex. This enables the computation of all the local $3$-profiles in parallel with minimal communication. 

We test out implementation in several experiments scaling up to $640$ cores on Amazon EC2. We find that our algorithm can estimate the $3$-profile of a graph in approximately the same time as triangle counting. 
For the harder problem of ego $3$-profiles, we introduce an algorithm that can estimate profiles of hundreds of thousands of vertices in parallel, in the timescale of minutes.
\end{abstract}

\section*{Categories and Subject Descriptors}
G.2.2 [\textbf{Graph Theory}]: Graph Algorithms; C.2.4 [\textbf{Distributed Systems}] Distributed Applications
\section*{Keywords}
3-profiles; Graph Sparsifiers; Motifs; Graph Engines; GraphLab; Distributed Systems; Graph Analytics

\section{Introduction}

Given a small integer $k$ (\textit{e.g.} $k=3$ or $4$), the $k$-profile of a graph $G(V,E)$ is a vector with one coordinate for each distinct $k$-node graph $H_i$ (see Figure \ref{fig:3profiles} for $k=3$). Each coordinate counts the number of times that $H_i$ appears as an induced subgraph of $G$. For example, the graph $G=K_4$ (the complete graph on $4$ vertices) has the $3$-profile $[0,0,0,4]$ since it contains $4$ triangles and no other (induced) subgraphs. The graph $C_5$ (the cycle on $5$ vertices, \textit{i.e.} a pentagon) has the $3$-profile $[0,5,5,0]$. 
Note that the sum of the $k$-profile is always $\binom{|V|}{k}$, the total number of subgraphs.

One can see $k$-profiles as a generalization of triangle (as well as other motif) counting problems. They are increasingly popular for graph analytics both for practical and theoretical reasons. They form a concise graph description that has found several applications for the web~\cite{becchetti08,OCallaghan2012}, social networks \cite{Ugander2013}, and biological networks~\cite{przPPIorig} and seem to be empirically useful. Theoretically, they connect to the emerging theory of graph homomorphisms, graph limits and graphons~\cite{Borgs2006,Ugander2013,lovasz2012large}.

\def\blockdist{1.5}
\def\edgedist{5.5}
\def\smalldist{.2}
\begin{figure}
\centering
\begin{tikzpicture}[
inner/.style={circle,draw,fill=black,inner sep=1pt, minimum size=1em},
outer/.style= {draw
}
]
\matrix (w0) [matrix of nodes, outer, nodes={inner}, label=$H_0$]{
     &{} \\
    {} & & {}\\
  };
  \matrix (w1) [matrix of nodes, outer, nodes={inner},right=\smalldist of w0, label=$H_1$]{
       &{} \\
      {} & & {}\\
    };
  \matrix (w2) [matrix of nodes, outer, nodes={inner},right=\smalldist of w1, label=$H_2$]{
       &{} \\
      {} & & {}\\
    };
  \matrix (w3) [matrix of nodes, outer, nodes={inner},right=\smalldist of w2, label=$H_3$]{
       &{} \\
      {} & & {}\\
    };
\draw[thick] (w1-2-1)--(w1-2-3);
  \draw[thick] (w2-2-3)--(w2-2-1)--(w2-1-2);
  \draw[thick] (w3-2-1)--(w3-1-2)--(w3-2-3)--(w3-2-1);
\end{tikzpicture}
\caption{Subgraphs in the $3$-profile of a graph. We call them (empty, edge, wedge, triangle). The $3$-profile of a graph $G$ counts how many times each of $H_i$ appears in $G$.}
\label{fig:3profiles}
\end{figure}
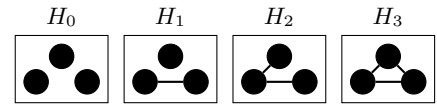

In this paper we introduce a novel distributed algorithm for estimating the $k=3$-profiles of massive graphs.
In addition to estimating the (global) $3$-profile, we address two more general problems. One is calculating the \textit{local} $3$-profile for each vertex $v_j$. This assigns a vector to each vertex that counts how many times $v_j$ participates in each subgraph $H_i$. These local vectors contain a higher resolution description of the graph and are used to obtain the global $3$-profile (simply by rescaled addition as we will discuss). 

The second related problem is that of calculating the \text{ego} $3$-profile for each vertex $v_j$. This is the $3$-profile of the graph $N(v_j)$ \textit{i.e.} the neighbors of $v_j$, also called the ego graph of $v_j$. The $3$-profile of the ego graph of $v_j$ can be seen as a projection of the vertex into a coordinate system~\cite{Ugander2013}. This is a very interesting idea of viewing a big graph as a collection of small dense graphs, in this case the ego graphs of the vertices. Note that calculating the ego $3$-profiles for a set of vertices of a graph is different (in fact, significantly harder) than calculating local $3$-profiles. 

\textbf{Contributions:} Our first contribution is a provable edge sub-sampling scheme: we establish sharp concentration results 
for estimating the entire $3$-profile of a graph. This allows us to randomly discard most edges of the graph and still have $3$-profile estimates that are provably within a bounded error with high probability. Our analysis is based on modeling the transformation from original to sampled graph as a one step Markov chain with transitions expressed as a function of the sampling probability. 
Our result is that a random sampling of edges forms a $3$-profile sparsifier, \textit{i.e.} a subgraph that preserves the elements of the $3$-profile with sufficient probability concentration. Our result is a generalization of the triangle sparsifiers by Tsourakakis \textit{et al. }~\cite{Tsourakakis2011sparsifier}. Our proof relies on a result by Kim and Vu~\cite{KimVu2000concentration} on concentration of multivariate polynomials, similarly to~\cite{Tsourakakis2011sparsifier}. Unfortunately, the Kim and Vu concentration holds only for a class of polynomials called totally positive and some terms in the $3$-profile do not satisfy this condition.
For that reason, the proof of~\cite{Tsourakakis2011sparsifier} does not directly extend beyond triangles. Our technical innovation involves showing that it is still possible to decompose our polynomials as combinations of totally positive polynomials using a sequence of variable changes. 

Our second innovation deals with designing an efficient, \textit{distributed} algorithm for estimating $3$-profiles on the sub-sampled graph. We rely on the Gather-Apply-Scatter model used in Graphlab PowerGraph~\cite{powergraphGAS2012} but, more generally, our algorithm fits the architecture of most graph engines. We introduce the concept of \textit{edge pivoting} which allows us to collect $2$-hop information without maintaining an explicit $2$-hop neighborhood list at each vertex. This enables the computation of all the local $3$-profiles in parallel. Each edge requires only information from its endpoints and each vertex only computes quantities using data from incident edges. For the problem of ego $3$-profiles, we show how to calculate them by combining edge pivot equations and local clique counts. 

We implemented our algorithm in GraphLab and performed several experiments scaling up to $640$ cores on Amazon EC2. 
We find that our algorithm can estimate the $3$-profile of a graph in approximately the same time as triangle counting. Specifically, we compare against the PowerGraph triangle counting routine and find that it takes us only $1\%$-$10\%$ more time to compute the full $3$-profile. For the significantly harder problem of ego $3$-profiles, we were able to compute (in parallel) the $3$-profiles of 
up to $100,000$ ego graphs in the timescale of several minutes. We compare our parallel ego $3$-profile algorithm to a simple sequential algorithm that operates on each ego graph sequentially and shows tremendous scalability benefits, as expected. 
Our datasets involve social network and web graphs with edges ranging in number from tens of millions to over one billion. We present results on both overall runtimes and network communication on multicore and distributed systems.

\section{Related Work}
In this section, we describe several related topics and discuss differences in relation to our work.

\noindent \textbf{Graph Sub-Sampling:} Random edge sub-sampling is a natural way to quickly obtain estimates for graph parameters.
For the case of triangle counting such graphs are called a triangle sparsifiers~\cite{Tsourakakis2011sparsifier}. Related ideas were explored in the Doulion algorithm~\cite{Tsourakakis,TsourakakisSubsamp2009,Tsourakakis2011sparsifier} with increasingly strong concentration bounds.
The recent work by Ahmed \textit{et al.}~\cite{Ahmed2014} develops subgraph estimators for clustering coefficient, triangle count, and wedge count in a streaming sub-sampled graph. Other recent work 
\cite{seshadri2012wedge,Jha2012Birthday,Seshadhri2013wedge,Bhuiyan2012,Jha2014} 
uses random sampling to estimate parts of the $3$ and $4$-profile. These methods do not account for a distributed computation model and require more complex sampling rules.
As discussed, our theoretical results build on \cite{Tsourakakis2011sparsifier} to define the first $3$-profile sparsifiers, sparse graphs that are \textit{a fortiori} triangle sparsifiers. 

\noindent \textbf{Triangle Counting in Graph Engines:} Graph engines (\textit{e.g.} Pregel, GraphLab, Galois, GraphX, see~\cite{Satish2014} for a comparison) are frameworks for expressing distributed computation on graphs in the language of vertex programs. Triangle counting algorithms \cite{Schank2007,becchetti08} form one of the standard graph analytics tasks for such frameworks~\cite{powergraphGAS2012,Satish2014}.
In \cite{Chu2011}, the authors list triangles efficiently, by partitioning the graph into components and processing each component in parallel.
Typically, it is much harder to perform graph analytics over the MapReduce framework but some recent work~\cite{Pagh2012,SuriReducer} has
used clever partitioning and provided theoretical guarantees for triangle counting.

\noindent \textbf{Matrix Formulations:} Fast matrix multiplication has been used for certain types of subgraph counting. Alon \textit{et al.} proposed a cycle counting algorithm which uses the trace of a matrix power on high degree vertices \cite{Alon1997}. 
Some of our edge pivot equations have appeared in \cite{Kloks2000eqs,Kowaluk2013eqs,Williams2014mod}, all in a centralized setting.
Related approximation schemes \cite{Tsourakakis} and randomized algorithms \cite{Williams2014mod} depend on centralized architectures and computing matrix powers of very large matrices.

\noindent \textbf{Frequent Subgraph Discovery:} The general problem of finding frequent subgraphs, also known as motifs or subgraph isomorphisms, is to find the number of occurrences of a small query graph within a larger graph. Typically frequent subgraph discovery algorithms offer pruning rules to eliminate false positives early in the search~\cite{Yan2002,Lee2012,Han2013,Ribeiro2010,Saltz2014}.
This is most applicable when subgraphs have labelled vertices or directed edges. For these problems, the number of unique isomorphisms grows much larger than in our application.

In \cite{Ugander2013}, subgraphs were queried on the ego graphs of users. While enumerating all $3$-sets and sampling $4$-sets neighbors can be done in parallel, forming the ego subgraphs requires checking for edges between neighbors. This suggests that a graph engine implementation would be highly preferable over an Apache Hive system. Our algorithms simultaneously compute the ego subgraphs and their profiles, reducing the amount of communication between nodes. 
Our algorithm is suitable for both NUMA multicore and distributed architectures, but our implementation focus in this paper is on GraphLab.

\noindent \textbf{Graphlets:} First described in \cite{przPPIorig}, graphlets generalize the concept of vertex degree to include the connected subgraphs a particular vertex participates in with its neighbors. Unique graphlets are defined at a vertex based on its degree in the subgraph. Graphlet frequency distributions (GFDs) have proven extremely useful in the field of bioinformatics. Specifically, GFD analysis of protein interaction networks helps to design improved generative models \cite{Hormozdiari2007}, accurate similarity measures \cite{przPPIorig}, 
and better features for classification 
\cite{Milenkovik2008,Shervashidze2009Learning}.
Systems that use our edge pivot equations (in a different form) appear in prior literature for calculating GFDs \cite{Hocevar2014bio,Marcus2012rage} but not for enabling distributed computation.

\section{Unbiased 3-profile Estimation}\label{sec:concentration}
In this section, we are interested in estimating the number of $3$ node subgraphs of type $H_0$, $H_1$, $H_2$ and $H_3$, as depicted in Figure \ref{fig:3profiles}, in a given graph $G$. Let the estimated counts be denoted $X_i,~i \in \{0,1,2,3\}$. Let the actual counts be $n_i,~i \in \{0,1,2,3\}$. This set of counts is called a $3$-profile of the graph, denoted with the following vector:
\begin{align}
\vecn_{3}(G) &= \left[n_0, \, n_1, \, n_2, \, n_3\right].
\end{align}
Because the vector is a scaled probability distribution, there are only $3$ degrees of freedom. Therefore, we calculate
\begin{align*}
\vecn_{3}(G) &= \left[{|V| \choose 3} - n_1 - n_2 - n_3, \, n_1, \, n_2, \, n_3\right].
\end{align*}

In the case of a large graph, computational difficulty in estimating the $3$-profile depends on the total number of edges in the large graph. We would like to estimate each of the $3$-profile counts within a multiplicative factor. So we first sub-sample the set of edges in the graph with probability $p$. We compute all $3$-profile counts of the sub-sampled graph exactly. Let $\{Y_i\}_{i=0}^{3}$ denote the exact $3$-profile counts of the random sub-sampled graph. 
We relate the sub-sampled $3$-profile counts to the original ones through a one step Markov chain involving transition probabilities. The sub-sampling process is the random step in the chain. Any specific subgraph is preserved with some probability and otherwise transitions to one of the other subgraphs. For example, a $3$-clique is preserved with probability $p^3$. Figure \ref{fig:3transition} illustrates the other transition probabilities.

\tikzstyle{sensor}=[draw, fill=blue!20, text width=5em, 
    text centered, minimum height=2.5em,drop shadow]
\tikzstyle{ann} = [above, text width=5em, text centered]
\tikzstyle{wa} = [sensor, text width=10em, fill=red!20, 
    minimum height=2.5em, rounded corners, drop shadow]
\tikzstyle{sc} = [sensor, text width=13em, fill=red!20, 
    minimum height=10em, rounded corners, drop shadow] 
\tikzstyle{bl} = [draw, fill=blue!20, rounded corners, drop shadow]
\tikzstyle{br} = [draw, fill=red!20, rounded corners, drop shadow]
\def\blockdist{1.5}
\def\edgedist{5.5}
\def\smalldist{.2}

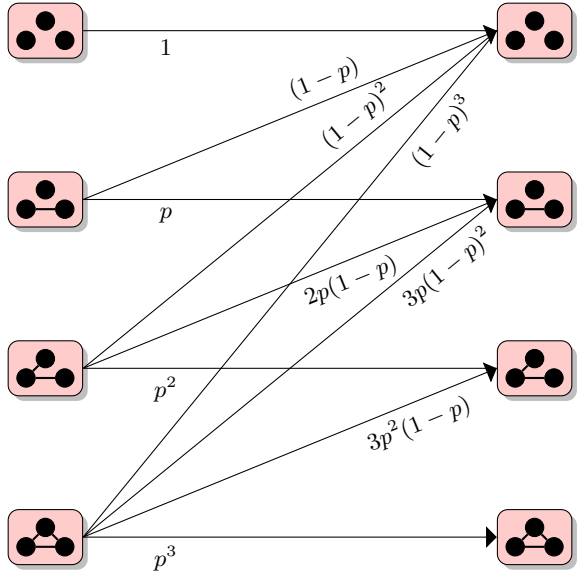
\begin{figure}
\centering
\begin{tikzpicture}[ 
  inner/.style={circle,draw,fill=black,inner sep=1pt, minimum size=.75em},
  outer/.style= {draw, fill=red!20, rounded corners, drop shadow}
  ]
  \matrix (w0) [matrix of nodes, outer, nodes={inner}]{
     &{} \\
    {} & & {}\\
  };
  \matrix (x0) [matrix of nodes, outer, nodes={inner}, right=\edgedist of w0]{
     &{} \\
    {} & & {}\\
  };
  \matrix (w1) [matrix of nodes, outer, nodes={inner}, below=\blockdist of w0]{
     &{} \\
    {} & & {}\\
  };
  \matrix (x1) [matrix of nodes, outer, nodes={inner}, right=\edgedist of w1]{
     &{} \\
    {} & & {}\\
  };
  \matrix (w2) [matrix of nodes, outer, nodes={inner}, below=\blockdist of w1]{
     &{} \\
    {} & & {}\\
  };
  \matrix (x2) [matrix of nodes, outer, nodes={inner}, right=\edgedist of w2]{
     &{} \\
    {} & & {}\\
  };
  \matrix (w3) [matrix of nodes, outer, nodes={inner}, below=\blockdist of w2]{
     &{} \\
    {} & & {}\\
  };
  \matrix (x3) [matrix of nodes, outer, nodes={inner}, right=\edgedist of w3]{
     &{} \\
    {} & & {}\\
  };

  \draw[thick] (w1-2-1)--(w1-2-3);
  \draw[thick] (x1-2-1)--(x1-2-3);
  \draw[thick] (w2-2-3)--(w2-2-1)--(w2-1-2);
  \draw[thick] (x2-2-3)--(x2-2-1)--(x2-1-2);
  \draw[thick] (w3-2-1)--(w3-1-2)--(w3-2-3)--(w3-2-1);
  \draw[thick] (x3-2-1)--(x3-1-2)--(x3-2-3)--(x3-2-1);

  \path [draw, -] (w0.east) -- node [below,pos=.2] {$1$} (x0.180) ;
  \path [draw, -triangle 90] (w1.east) -- node [above right, sloped] {$(1-p)$} (x0.180);
  \path [draw, -] (w2.east) -- node [above right, sloped, pos=.6] {$(1-p)^2$} (x0.180);
  \path [draw, -] (w3.east) -- node [below right, sloped, near end] {$(1-p)^3$} (x0.180);
  \path [draw, -] (w1.east) -- node [below,pos=.2] { $p$} (x1.180);
  \path [draw, -triangle 90] (w2.east) -- node [below right, sloped, pos=.5] {$2p(1-p)$} (x1.180);
  \path [draw, -] (w3.east) -- node [below right, sloped, pos=.73] {$3p(1-p)^2$} (x1.180);
  \path [draw, -] (w2.east) -- node [below,pos=.2] {$p^2$} (x2.180);
  \path [draw, -triangle 90] (w3.east) -- node [below right, sloped, pos=.65] {$3p^2(1-p)$} (x2.180);
  \path [draw, -triangle 90] (w3.east) -- node [below,pos=.2] {$p^3$} (x3.180);
\end{tikzpicture}
\caption{Edge sampling process.}
\label{fig:3transition}
\end{figure}

In expectation, this yields the following linear system:
\begin{align}
\begin{bmatrix}
\E{Y_0} \\ \E{Y_1} \\ \E{Y_2} \\ \E{Y_3}
\end{bmatrix}
= \begin{bmatrix}
1 & 1-p & (1-p)^2 & (1-p)^3 \\
0 & p & 2p(1-p) & 3p(1-p)^2 \\
0 & 0 & p^2 & 3p^2(1-p) \\
0 & 0 & 0 & p^3
\end{bmatrix}
\begin{bmatrix}
n_0 \\ n_1 \\ n_2 \\ n_3
\end{bmatrix} ,\label{eq:channelsys}
\end{align}

from which we obtain unbiased estimators for each entry in $\mathbf{X}(G) = \left[X_0, \, X_1, \, X_2, \, X_3\right]$:
\begin{align}
X_0 &= Y_0 - \frac{1-p}{p}Y_1 + \frac{(1-p)^2}{p^2}Y_2 - \frac{(1-p)^3}{p^3}Y_3 \label{eq:estimatorStart} \\
X_1 &= \frac{1}{p}Y_1 - \frac{2(1-p)}{p^2}Y_2 + \frac{3(1-p)^2}{p^3}Y_3 \label{eqn:estdisc}\\
X_2 &= \frac{1}{p^2}Y_2 - \frac{3(1-p)}{p^3}Y_3 \label{eq:estimatorWedge} \\
X_3 &= \frac{1}{p^3}Y_3 . \label{eq:estimatorEnd}
\end{align}

\begin{lem}
$\mathbf{X}(G)$ is an unbiased estimator of $\vecn(G)$.
\end{lem}
\begin{proof}
By substituting \eqref{eq:estimatorStart}--\eqref{eq:estimatorEnd} into \eqref{eq:channelsys}, clearly $\E{X_i} = n_i$ for $i = 0,1,2,3$.
\end{proof}

We now turn to prove concentration bounds for the above estimators. We introduce some notation for this purpose. Let $X$ be a real polynomial function of $m$ real random variables $\{t_i\}_{i=1}^m$.
Let $\vecalpha = (\alpha_1, \alpha_2, \ldots, \alpha_m) \in \mathbb{Z}_+^m$ and define $\mathbb{E}_{\geq 1}[X] = \max_{\vecalpha: \pnorm{\vecalpha}{1} \geq 1} \mathbb{E}(\partial^\vecalpha X)$, where
\begin{align}
\mathbb{E}(\partial^\vecalpha X) = \mathbb{E} \left[(\frac{\partial}{\partial t_1})^{\alpha_1} \ldots (\frac{\partial}{\partial t_m})^{\alpha_m} \left[ X(t_1,\ldots, t_m) \right] \right].
\end{align}
Further, we call a polynomial totally positive if the coefficients of all the monomials involved are non-negative.
We state the main technical tool we use to obtain our concentration results.
\begin{prop}[Kim-Vu Concentration \cite{KimVu2000concentration}] \label{KIMVUCONCENTRATION}
Let $X$ be a random totally positive Boolean polynomial in $m$ Boolean random variables with degree at most $k$. If $\mathbb{E}[X] \geq \mathbb{E}_{\geq 1}[X]$, then
\begin{align}
\mathbb{P}\left( |X - \mathbb{E}[X] | > a_k \sqrt{\mathbb{E}[X] \mathbb{E}_{\geq 1}[X]}  \lambda^k \right) \nonumber \\
= \mathcal{O}\left(\exp \left(-\lambda + (k-1) \log m \right) \right)   \label{eq:kvtheorem}
\end{align}
 for any $\lambda > 1$, where $a_k = 8^k k!^{1/2}$.
\end{prop}

The above proposition was used to analyze $3$-profiles of Erd\H{o}s-R\'enyi random ensembles ($G_{n,p}$) in \cite{KimVu2000concentration}. Later, this was used to derive concentration bounds for triangle sparsifiers in \cite{Tsourakakis2011sparsifier}. 
Here, we extend \S 4.3 of \cite{KimVu2000concentration} to the $3$-profile estimation process, on an arbitrary edge-sampled graph.

\begin{theorem}\label{PROFILECONCENTRATION}
(\textit{Generalization of triangle sparsifiers to $3$-profile sparsifiers}) Let $\vecn(G)=\left[n_0, \, n_1, \, n_2, \, n_3\right]$ be the $3$-profile of a graph $G(V,E)$. Let $\lvert V \rvert=n$ and $\lvert E \rvert=m$. Let $\hat{\vecn}(G) = \left[Y_0, \, Y_1, \, Y_2, \, Y_3\right]$ be the $3$-profile of the subgraph obtained by sampling each edge in $G$ with probability $p$. Let $\alpha,\beta$ and $\Delta$ be the largest collection of $H_1$'s, wedges and triangles that share a common edge. Define $\mathbf{X}(G)$ according to \eqref{eq:estimatorStart}--\eqref{eq:estimatorEnd}, $\epsilon > 0$, and $\gamma>0$. If $p,\epsilon$ satisfy:
\begin{align}\label{eqn:thmconds}
\frac{n_0}{3 \max\{\alpha,\beta,\Delta\}} &\geq \frac{a_3^2 \log^6 \left(m^{2+\gamma} \right)}{\epsilon^2} \nonumber \\
\frac{p} { \max \{\frac{1}{\sqrt[3]{n_3}},  \Delta/n_3 \} } &\geq \frac{a_3^2 \log^6 \left(m^{2+\gamma} \right)}{\epsilon^2} \nonumber \\
\frac{p} {\alpha/n_1 } &\geq \frac{a_1^2 \log^2 \left(m^{\gamma} \right)}{\epsilon^2} \nonumber \\
\frac{p}{\max \{\frac{\beta}{n_2},\frac{1}{\sqrt{n_2}} \}}  &\geq  \frac{a_2^2 \log^4 \left(m^{1+\gamma} \right)}{\epsilon^2} ,
\end{align}
then $\pnorm{\mathbf{X}(G)-\vecn(G)}{\infty} \leq 12\epsilon {|V| \choose 3}$ with probability at least $1 - \frac{1}{m^{\gamma}}$.
\end{theorem}
\begin{proof}
Full proof can be found in the appendix. Note that, as we mentioned, this proof uses a new polynomial decomposition technique to apply the Kim-Vu concentration.
\end{proof}

The sampling probablilty $p$ in Theorem \ref{PROFILECONCENTRATION} depends poly-logarithmically on the number of edges and linearly on the fraction of each subgraph which occurs on a common edge. For example, if all of the wedges in $G$ depend on a single edge, \textit{i.e.} $\beta = n_2$, then the last equation suggests the presence of that particular edge in the sampled graph will dominate the overall sparsifier quality.

\section{Local 3-profile Calculation}\label{sec:local3}
In this section, we describe how to obtain two types of $3$-profiles for a given graph $G$ in a deterministic manner. These algorithms are distributed
and can be applied independently of the edge sampling described in Section \ref{sec:concentration}.

The key to our approach is to identify subgraphs at a vertex based on the degree with which it participates in the subgraph. From the perspective of a given vertex $v$, there are actually six distinct $3$ node subgraphs up to isomorphism as given in Figure \ref{fig:local3prof}.
Let $n_{0,v},n_{1,v}^{e},n_{1,v}^d,n_{2,v}^e,n_{2,v}^d$, and $n_{3,v}$ denote the corresponding local subgraph counts at $v$. We will first outline an approach that calculates these counts and then add across different vertex perspectives to calculate the final $4$ scalars ($n_{2,v}=n_{2,v}^e+n_{2,v}^c$ and $n_{1,v}=n_{1,v}^e+n_{1,v}^d$). It is easy to see that the global counts can be obtained from these local counts by summing across vertices:  
  \begin{align}
    n_i= \frac{1}{3} \left( \sum \limits_{v \in V} n_{i,v}\right), ~ \forall i .
   \end{align} 
   
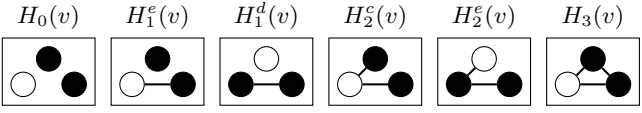
\begin{figure}[t]
\centering
\begin{tikzpicture}[
inner/.style={circle,draw,fill=black,inner sep=1pt, minimum size=1em},
outer/.style= {draw
}
]
  \matrix (w0) [matrix of nodes, outer, nodes={inner}, label=$H_0(v)$]{
         &{} \\
        |[fill=white]| {} & & {}\\
      };
    \matrix (w1) [matrix of nodes, outer, nodes={inner}, right=\smalldist of w0, label=$H_1^{e}(v)$]{
       &{} \\
      |[fill=white]| {} & & {}\\
    };
  \matrix (w1d) [matrix of nodes, outer, nodes={inner},right=\smalldist of w1, label=$H_1^{d}(v)$]{
         & |[fill=white]|{} \\
         {} & & {}\\
      };
   \matrix (w2) [matrix of nodes, outer, nodes={inner},right=\smalldist of w1d, label=$H_2^c(v)$]{
       &{} \\
      |[fill=white]|{} & & {}\\
    };
    \matrix (w2e) [matrix of nodes, outer, nodes={inner},right=\smalldist of w2, label=$H_2^{e}(v)$]{
        & |[fill=white]|{} \\ 
           {} & & {}\\
        };
  \matrix (w3) [matrix of nodes, outer, nodes={inner},right=\smalldist of w2e, label=$H_3(v)$]{
       &{} \\
	      |[fill=white]| {} & & {}\\
    };
\draw[thick] (w1-2-1)--(w1-2-3);
\draw[thick] (w1d-2-1)--(w1d-2-3);
  \draw[thick] (w2-2-3)--(w2-2-1)--(w2-1-2);
    \draw[thick] (w2e-2-3)--(w2e-2-1)--(w2e-1-2);
  \draw[thick] (w3-2-1)--(w3-1-2)--(w3-2-3)--(w3-2-1);
\end{tikzpicture}
\caption{Unique $3$-subgraphs from vertex perspective (white vertex corresponds to $v$).}
\label{fig:local3prof}
\end{figure}

\subsection{Distributed Local 3-profile}
We will now give our approach for calculating the local $3$-profile counts of $G(V,E)$ using only local information combined with $|V|$ and $|E|$.

\noindent \textbf{Scatter:}
We assume that every edge $(v,a)$ has access to the neighborhood sets of both $v$ and $a$, i.e. $\Gamma(v)$ and $\Gamma(a)$. Therefore, intersection sizes are first calculated at every edge, \textit{i.e.} $|\Gamma(v) \cap \Gamma(a)|$. Each edge computes the following scalars and stores them: 
\begin{align}\label{eq:scat1}
n_{3,va} &= |\Gamma(v) \cap \Gamma(a)|,~ n_{2,va}^e = |\Gamma(v)| - |\Gamma(v) \cap \Gamma(a)| -1. \nonumber \\
n_{2,va}^c &= |\Gamma(a)| - |\Gamma(v) \cap \Gamma(a)| - 1. \nonumber \\
n_{1,va} &= |V| - (|\Gamma(v)| + |\Gamma(a)| - |\Gamma(v)\cap \Gamma(a)|).
\end{align}
 The computational effort at every edge is at most $\mathcal{O}(d_{\mathrm{max}})$, where $d_{\mathrm{max}}$ is the maximum degree of the graph, for the neighborhood intersection size.

\noindent \textbf{Gather:}
In the next round, vertex $v$ ``gathers''  the above scalars in the following way:
\begin{align}
\begin{split}
\label{eq:app2}
n_{2,v}^e &= \sum_{a \in \Gamma(v)} n_{2,va}^e , ~~
n_{1,v}^e = \sum_{a \in \Gamma(v)} n_{1,va}. \\
n_{2,v}^c & \overset{a}= \frac{1}{2}\sum_{a \in \Gamma(v)}  n_{2,va}^c, ~~ 
n_{3,v} \overset{b}= \frac{1}{2} \sum_{a \in \Gamma(v)} n_{3,va}.  \\
n_{1,v}^d & \overset{c}= |E| - |\Gamma(v)| - n_{3,v} - n_{2,v}^e .\\
n_{0,v} &= \binom{|V| - 1}{2} - n_{1,v}-n_{2,v}-n_{3,v}.
\end{split}
\end{align}
Here, relations (a) and (b) are because triangles and wedges from center are double counted. (c) comes from noticing that each triangle and wedge from endpoint excludes an extra edge from forming $H_1^d(v)$. In this gather stage, the communication complexity is $\mathcal{O}(M)$ where it is assumed that $\Gamma(v)$ is stored over $M$ different machines. The corresponding distributed algorithm is described in Algorithm \ref{alg:3plocalGAS}.

\subsection{Distributed Ego 3-profile}
In this section, we give an approach to compute ego $3$-profiles for a set of vertices ${\cal V} \subseteq V$ in $G$. 
For each vertex $v$, the algorithm returns a $3$-profile corresponding to that vertex's ego $N(v)$, a subgraph induced by the neighborhood set $\Gamma(v)$, including edges between neighbors and excluding $v$ itself. Formally, our goal is to compute $\{ \mathbf{n} (N(v)) \}_{v \in {\cal V}}$. 
Clearly, this can be accomplished in two steps repeated serially on all $v \in {\cal V}$: first obtain the ego subgraph $N(v)$ and then pass as input to Algorithm \ref{alg:3plocalGAS}, summing over the ego vertices $\Gamma(v)$ to get a global count.
The serial implementation is provided in Algorithm \ref{alg:3egoSerial}. We note that this was essentially done in \cite{Ugander2013}, where ego subgraphs were extracted from a common graph separately from $3$-profile computations.

\begin{figure}
\centering
\begin{tikzpicture}[
inner/.style={circle,draw,fill=black,inner sep=1pt, minimum size=1em},
outer/.style= {draw
}
]
	\matrix (w6) [matrix of nodes, outer, nodes={inner}, label=below:$F_0(v)$, column sep=10pt, row sep=10pt]{
	{} &{} \\ 
	|[fill=white]|{} & {}\\
	};
	\matrix (w8) [matrix of nodes, outer, nodes={inner},right=\smalldist of w6, label=below:$F_1(v)$, column sep=10pt, row sep=10pt]{
	{} &{} \\ 
	|[fill=white]|{} & {}\\
	};
	\matrix (w9) [matrix of nodes, outer, nodes={inner},right=\smalldist of w8, label=below:$F_2(v)$, column sep=10pt, row sep=10pt]{
	{} &{} \\ 
	|[fill=white]|{} & {}\\
	};
	\matrix (w10) [matrix of nodes, outer, nodes={inner},right=\smalldist of w9, label=below:$F_3(v)$, column sep=10pt, row sep=10pt]{
	{} &{} \\ 
	|[fill=white]|{} & {}\\
	};
\draw[thick] (w6-1-1)--(w6-2-1)--(w6-2-2);
\draw[thick] (w6-1-2)--(w6-2-1);
\draw[thick] (w8-1-1)--(w8-2-1)--(w8-2-2)--(w8-1-2)--(w8-2-1);
\draw[thick] (w9-1-1)--(w9-2-1)--(w9-2-2)--(w9-1-2)--(w9-2-1)--(w9-1-1)--(w9-1-2);
\draw[thick] (w10-1-1)--(w10-2-1)--(w10-2-2)--(w10-1-2)--(w10-1-1)--(w10-2-2)--(w10-1-2)--(w10-2-1);
\end{tikzpicture}
\caption{$4$-subgraphs for Ego $3$-profiles (white vertex corresponds to $v$).} \label{fig:4profiles}
\end{figure}
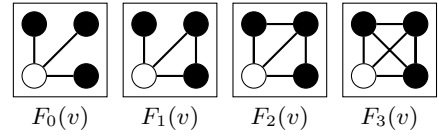

Instead, Algorithm \ref{alg:3egoParallel} provides a parallel implementation which solves the problem by finding cliques in parallel for all $v\in {\cal V}$.
The main idea behind this approach is to realize that calculating the $3$-profile on the induced subgraph $N(v)$ is exactly equivalent to computing specific $4$-node subgraph frequencies among $v$ and $3$ of its neighbors, enumerated as $F_i(v), ~0 \leq i \leq 3$ in Figure \ref{fig:4profiles}. Now, the aim is to calculate $F_i(v)$'s, effectively part of a local $4$-profile.

\noindent \textbf{Scatter:}
We assume that every edge $(v,a)$ has already computed the scalars 
from (\ref{eq:scat1}). Additionally, every edge $(v,a)$ also computes the list ${\cal N}_{va}=\Gamma(v) \cap \Gamma(a)$ instead of only its size. The computational complexity is still $\mathcal{O}(d_{\mathrm{max}})$. 

\noindent \textbf{Gather:}
First, the vertex ``gathers'' the following scalars, forming three \textit{edge pivot equations} in unknown variables $F_i(v)$:
\begin{align}
\begin{split}
\sum_{a \in \Gamma(v)} {n_{2,va}^c \choose 2}  &= 3F_0(v) + F_1(v)  \\ 
\sum_{a \in \Gamma(v)} {n_{3,va} \choose 2} &= F_2(v) + 3F_3(v) \\
\sum_{a \in \Gamma(v)} n_{2,va}^c n_{3,va} &= 2F_1(v) + 2F_2(v) \label{eq:trisq} 
\end{split}
\end{align}

By choosing two subgraphs that the edge $(v,a)$ participates in, and then summing over neighbors $a$, these equations gather implicit connectivity information $2$ hops away from $v$.
However, note that there are only three equations in four variables and we must count one of them directly, namely the number of $4$-cliques $F_3(v)$. Therefore, at the same gather step, the vertex also creates the list ${\cal CN}_v=\bigcup_{a \in \Gamma(v), p \in {\cal N}_{va}} (a,p)$.  
Essentially, this is the list of edges in the subgraph induced by $\Gamma(v)$. 
This requires worst case communication proportional to the number of edges in $N(v)$, independent of the number of machines $M$.

\noindent \textbf{Scatter:}
Now, at the next scatter stage, each edge $(v,a)$ accesses the pair of lists ${\cal CN}_v,{\cal CN}_a$. Each edge $(v,a)$ computes 
the number of $4$-cliques it is a part of, defined as follows: 
  \begin{equation}\label{eq:n4edge}
     n_{4,va}= \sum \limits_{i,j \in \Gamma(v) \cap \Gamma(a) } \mathbf{1} \left( (i,j) \in {\cal CN}_v \right).
  \end{equation}
This incurs a computation time of $|{\cal CN}_v|$. 

\noindent \textbf{Gather:}
In the final gather stage, every vertex $v$ accumulates these scalars to get $F_{3}(v)= \frac{1}{3} \sum \limits_{a \in \Gamma(v)} n_{4,va}$ requiring $\mathcal{O}(M)$ communication time. As in the previous section, the scaling accounts for extra counting. Finally, the vertex solves the equations \eqref{eq:trisq} using $F_3(v)$.

\section{Implementation and Results}
In this section, we describe the implementation and the experimental results of the $\prof$, $\egopar$ and $\egoser$ algorithms. We implement our algorithms on GraphLab v2.2 (PowerGraph) \cite{powergraphGAS2012}. The performance (running time and network usage) of our $\prof$ algorithm is compared with the Undirected Triangles Count Per Vertex (hereinafter referred to as $\trian$) algorithm shipped with GraphLab. We show that in time and network usage comparable to the built-in $\trian$ algorithm, our $\prof$ can calculate all the local and global $3$-profiles. Then, we compare our parallel implementation of the ego $3$-profile algorithm, $\egopar$, with the naive serial implementation, $\egoser$. It appears that our parallel approach is much more efficient and scales much better than the serial algorithm. The sampling approach, introduced for the $\prof$ algorithm, yields promising results -- reduced running time and network usage while still providing excellent accuracy. We support our findings with several experiments over various data sets and systems.

\noindent \textbf{Vertex Programs:}
Our algorithms are implemented using a standard GAS (gather, apply, scatter) model \cite{powergraphGAS2012}. We implement the three functions \texttt{gather()}, \texttt{apply()}, and \texttt{scatter()} to be executed by each vertex. Then we signal subsets of vertices to run in a specific order.

\begin{algorithm}[ht]
    \caption{\prof}
   \label{alg:3plocalGAS}
\begin{algorithmic}
    \STATE {\bfseries Input:} Graph $G(V,E)$ with $|V|$ vertices, $|E|$ edges
    \STATE {\bfseries Gather:} For each vertex $v$ union over edges of the `other' vertex in the edge, $\cup_{a \in \Gamma(v)} a = \Gamma(v)$.
	\STATE {\bfseries Apply:} Store the gather as vertex data $\texttt{v.nb}$, size automatically stored.
	\STATE {\bfseries Scatter:} For each edge $e_{va}$, compute and store scalars in \eqref{eq:scat1}.
	\STATE {\bfseries Gather:} For each vertex $v$, sum edge scalar data of neighbors 
	\INDSTATE[1] $\texttt{g} \gets \sum_{(v,a) \in \Gamma(v)} \texttt{e.data}$.
	\STATE {\bfseries Apply:} For each vertex $v$, calculate and store the quantities described in \eqref{eq:app2}.
    \RETURN [\texttt{v: v.n0 v.n1 v.n2 v.n3}]
\end{algorithmic}
\end{algorithm}

\begin{algorithm}[ht]
   \caption{\egoser}
   \label{alg:3egoSerial}
\begin{algorithmic}
   \STATE {\bfseries Input:} Graph $G(V,E)$ with $|V|$ vertices, $|E|$ edges, set of ego vertices ${\cal V}$
	\FOR{$v \in {\cal V}$} 
		\STATE Signal $v$ and its neighbors.
		\STATE Include an edge if both its endpoints are signaled.
		\STATE Run Algorithm \ref{alg:3plocalGAS} on the graph induced by the neighbors and edges between them.
	\ENDFOR
	 \RETURN [\texttt{v: vego.n0 vego.n1 vego.n2 vego.n3}]
\end{algorithmic}
\end{algorithm}

\begin{algorithm}[ht]
   \caption{\egopar}
   \label{alg:3egoParallel}
\begin{algorithmic}
   \STATE {\bfseries Input:} Graph $G(V,E)$ with $|V|$ vertices, $|E|$ edges, set of ego vertices ${\cal V}$
		\STATE {\bfseries Gather:} For each vertex $v$ union over edges of the `other' vertex in the edge, $\cup_{e_{va}} a = \Gamma(v)$.
		\STATE {\bfseries Apply:} Store the gather as vertex data $\texttt{v.nb}$, size automatically stored.
		\STATE {\bfseries Scatter:} For each edge $e_{va}$, compute and store as edge data:
			 \INDSTATE[1] Scalars in \eqref{eq:scat1}.
			 \INDSTATE[1] The list ${\cal N}_{va}$.
		\STATE {\bfseries Gather:} For each vertex $v$, sum edge data of neighbors: 
		     \INDSTATE[1]  Acumulate LHS of \eqref{eq:trisq}. 
		     \INDSTATE[1] $\texttt{g}.{\cal CN} \gets \texttt{g}.{\cal CN} \cup {\cal N}_{va}$.
		\STATE {\bfseries Apply:} Obtain ${\cal CN}_v$ and equations in \eqref{eq:trisq} using the scalars and $\texttt{g}.{\cal CN}$.
		\STATE {\bfseries Scatter:} Scatter ${\cal CN}_v,{\cal CN}_a$ to all edges $(v,a)$.
					 \INDSTATE[1] Compute $n_{4,va}$ as in \eqref{eq:n4edge}.
		\STATE {\bfseries Gather:} Sum edge data $n_{4,va}$ of neighbors at $v$.
		\STATE {\bfseries Apply:} Compute $F_3(v)$.
   \RETURN [\texttt{v: vego.n0 vego.n1 vego.n2 vego.n3}]
\end{algorithmic}
\end{algorithm}

\noindent \textbf{The Systems:}
We perform the experiments on three systems. The first system is a single power server, further referred to as Asterix. The server is equipped with 256 GB of RAM and two Intel Xeon E5-2699 v3 CPUs, 18 cores each. Since each core has two hardware threads, up to 72 logical cores are available to the GraphLab engine.

The next two systems are EC2 clusters on AWS (Amazon Web Services) \cite{aws}. One is comprised of 12 m3.2xlarge machines, each having 30 GB RAM and 8 virtual CPUs. Another system is a cluster of 20 c3.8xlarge machines, each having 60 GB RAM and 32 virtual CPUs.

\noindent \textbf{The Data:}
In our experiments we used five real graphs. These graphs represent different datasets: social networks (LiveJournal and Twitter), citations (DBLP), knowledge content (Wikipedia), and WWW structure (PLD -- pay level domains). Graph sizes are summarized in Table \ref{tab:datsets}.

\begin{table}[h]
\centering
\caption{Datasets} \label{tab:datsets}
\begin{tabular}{|l|r|r|}
\hline
Name & Vertices & Edges (undirected)  \\ \hline
Twitter \cite{twitterKwak10www} & $41,652,230$ & $1,202,513,046$ \\
PLD \cite{pldgraph} & $39,497,204$ & $582,567,291$ \\
LiveJournal \cite{snapnets} & $4,846,609$ & $42,851,237$ \\
Wikipedia \cite{wikigraph} & $3,515,067$ & $42,375,912$ \\
DBLP \cite{snapnets} & $317,080$ & $1,049,866$ \\
\hline
\end{tabular}
\end{table}

\subsection{Results}
Experimental results are averaged over $3-10$ runs.

\noindent \textbf{Local 3-profile vs. triangle count:}
The first result is that our $\prof$ is able to compute all the local $3$-profiles in almost the same time as the GraphLab's built-in $\trian$ computes the local triangles (i.e., number of triangles including each vertex). Let us start with the first AWS cluster with less powerful machines (m3.x2large). In Figure \ref{fig:awsm3} (a) we can see that for the LiveJornal graph, for each sampling probability $p$ and for each number of nodes (i.e., machines in the cluster), $\prof$ achieves running times comparable to $\trian$. Notice also the benefit in running time achieved by sampling. We can reduce running time almost by half, without significantly sacrificing accuracy (which will be discussed shortly). While the running time is decreased as the number of nodes grows (more computing resources become available), the network usage becomes higher (see Figure \ref{fig:awsm3} (c)) due to the extensive inter-machine communication in GraphLab. We can also see that sampling can significantly reduce network usage. In Figures \ref{fig:awsm3} (b) and (d), we can see similar behavior for the Wikipedia graph: running time and network usage of $\prof$ is comparable to $\trian$.

Next, we conduct the experiments on the second AWS cluster with more powerful (c3.8xlarge) machines. For LiveJournal, we note modest improvements in running time for nearly the same network bandwidth observed in Figure \ref{fig:awsm3}. On this system we were able to run $\prof$ and $\trian$ on the much larger PLD graph. In Figures \ref{fig:awsc3} (b) and (d) we compare the running time and network usage of both algorithms. For the large PLD graph, the benefit of sampling can be seen clearly; by setting $p=0.1$, the running time of $\prof$ is reduced by a factor of $4$ and the network usage is reduced by a factor of $2$. Figure \ref{fig:aws20node} shows the performance of $\prof$ and $\trian$ on the LiveJournal and Wikipedia graphs. 
We can see that the behavior of running times and the network usage of the $\prof$ algorithm is consistently comparable to $\trian$ across the various graphs, sampling, and system parameters.

Let us now show results of the experiments performed on a single powerful machine (Asterix). Figure \ref{fig:asterixRuntimes} (a) shows the running times for $\prof$ and $\trian$ for Twitter and PLD graphs. We can see that on the largest graph in our dataset (Twitter), the running time of $\prof$ is less than $5\%$ larger than that of $\trian$, and for the PLD graph the difference is less than $3\%$ (for $p=1$). Twitter takes roughly twice as long to compute as PLD, implying that these algorithms have running time proportional to the graph's number of edges.

Finally, we show that while the sampling approach can significantly reduce the running time and network usage, it has negligible effect on the accuracy of the solution. Notice that the sampling accuracy refers to the global $3$-profile count (\textit{i.e.}, the sum of all the local $3$-profiles over all vertices in a graph).
In Figure \ref{fig:3profAcc} we show accuracy of each scalar in the $3$-profile. For the accuracy metrics, we use ratio between the exact count (obtained running $\prof$ with $p=1$) divided by the estimated count (i.e., the output of our $\prof$ when $p<1$). It can be seen that for the three graphs, all the $3$-profiles are very close to $1$. E.g., for the PLD graph, even when $p=0.01$, the accuracy is within $0.004$ from the ideal value of $1$. Error bars mark one standard deviation from the mean, and across all graphs the largest standard deviation is $0.031$. As $p$ decreases, the triangle estimator suffers the greatest loss in both accuracy and consistency.

\begin{figure}[h]
	\centering
	\subfloat[]{\includegraphics[width=0.5\columnwidth]{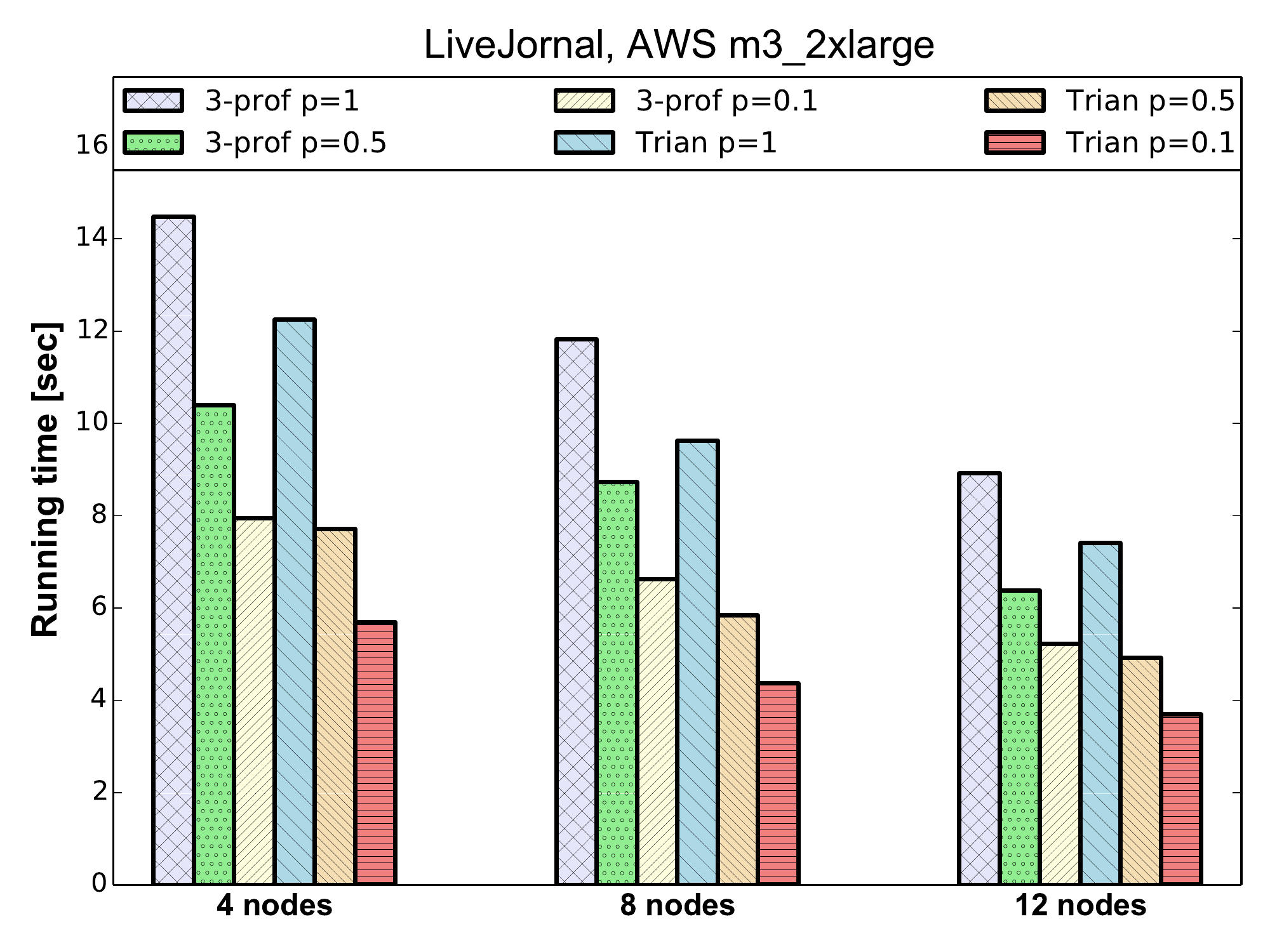}} 
	\subfloat[]{\includegraphics[width=0.5\columnwidth]{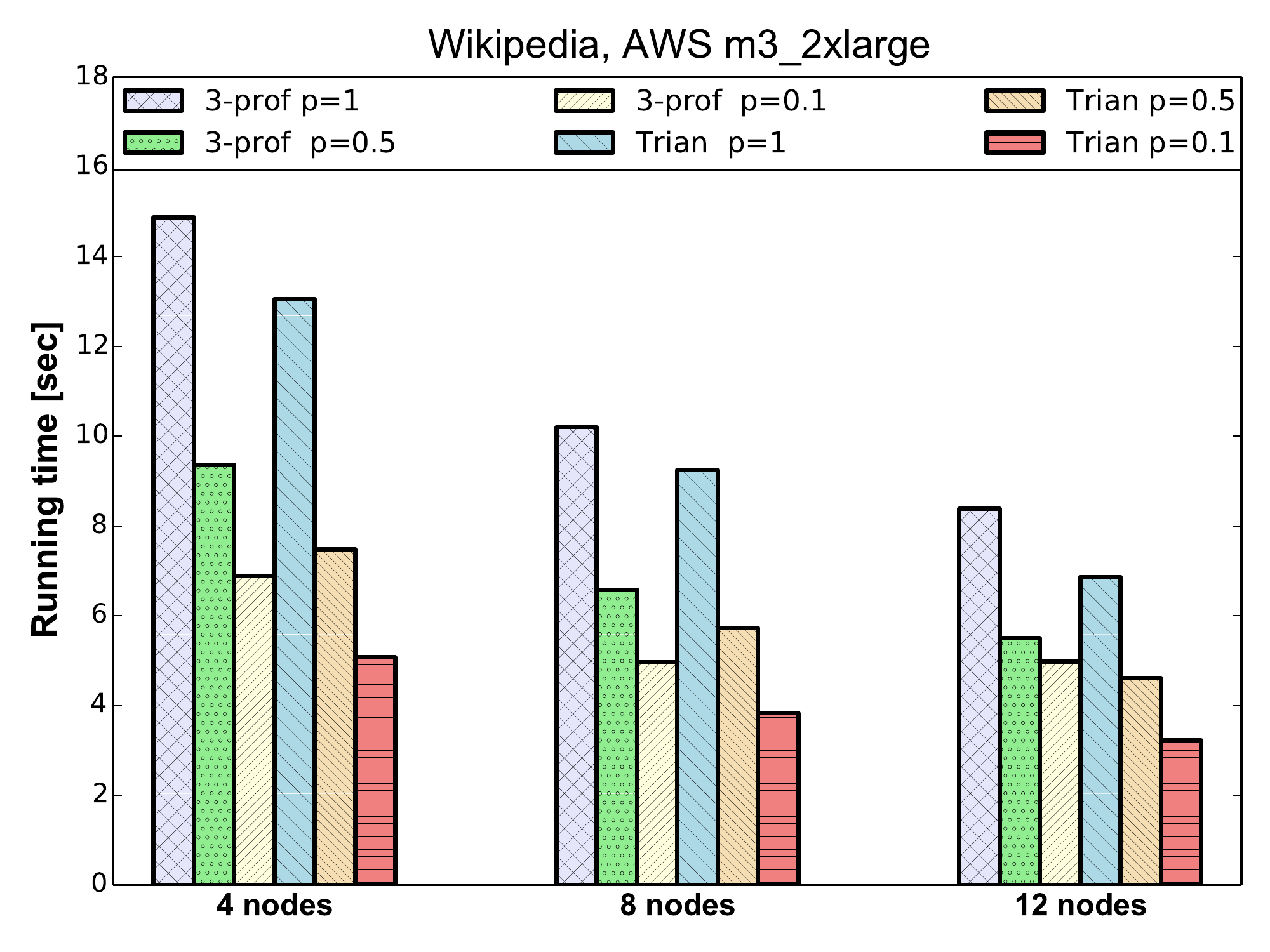}} \\ 
	\subfloat[]{\includegraphics[width=0.5\columnwidth]{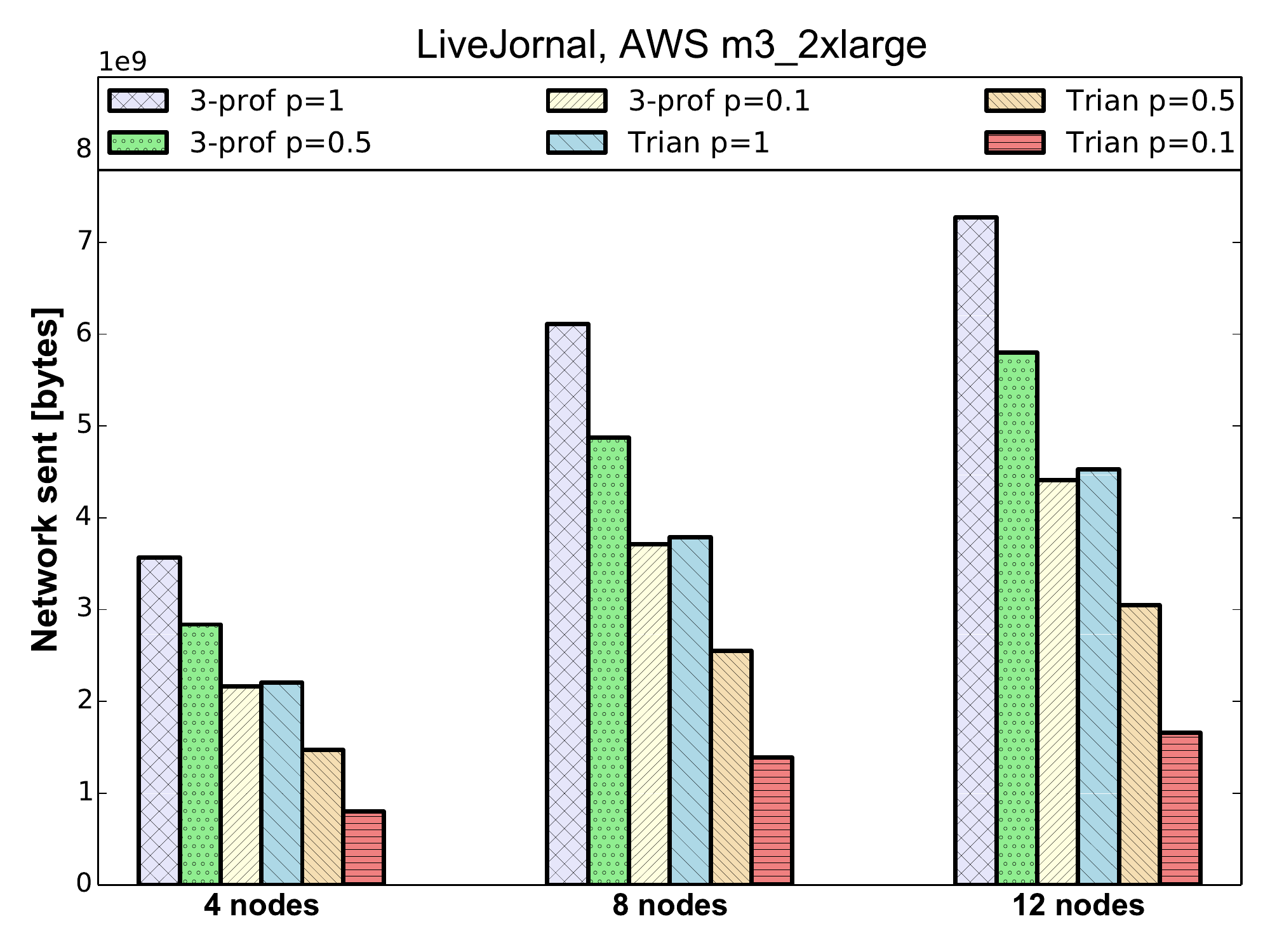}} 
	\subfloat[]{\includegraphics[width=0.5\columnwidth]{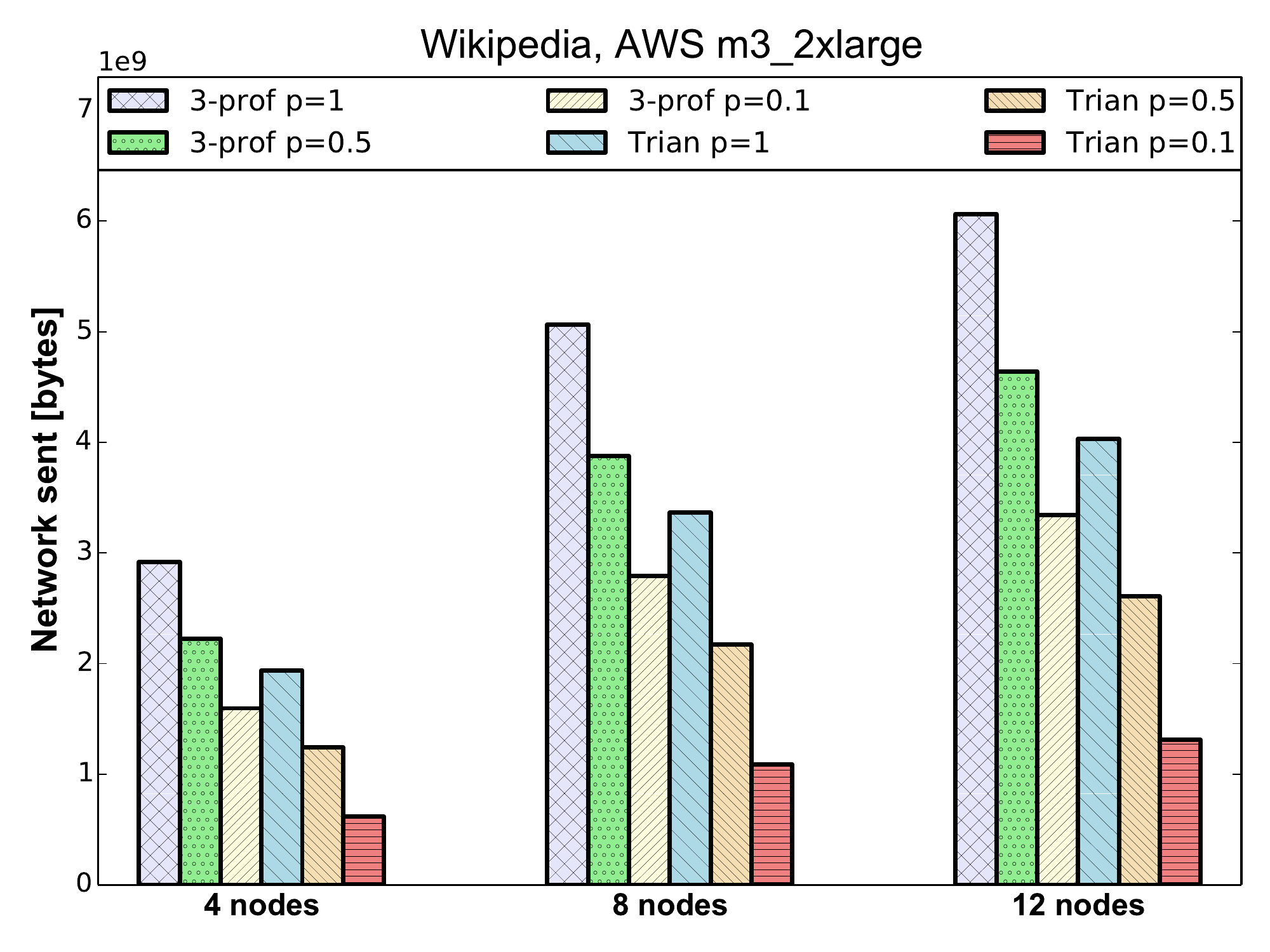}} \\ 
    \caption{\small AWS \texttt{m3\_2xlarge} cluster. $\prof$ vs. $\trian$ algorithms for LiveJournal and Wikipedia datasets (average of $3$ runs). $\prof$ achieves comparable performance to triangle counting.
    (a,b) -- Running time for various numbers of nodes (machines) and various sampling probabilities $p$.
    (c,d) -- Network bytes sent by the algorithms for various numbers of nodes and various sampling probabilities $p$.
    }
    \label{fig:awsm3}
\end{figure}

\begin{figure}[h]
	\centering
	\subfloat[]{\includegraphics[width=0.5\columnwidth]{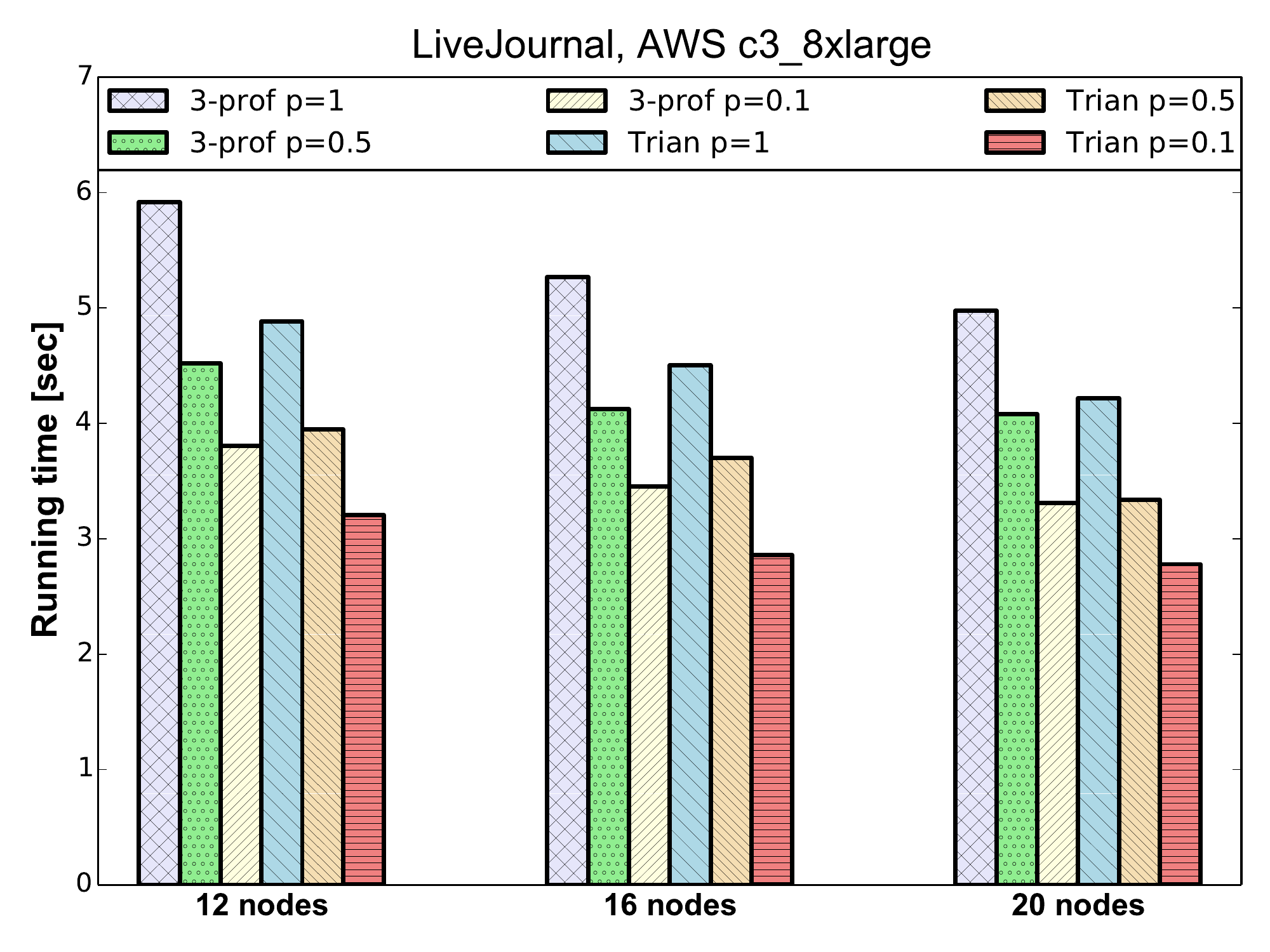}} 
	\subfloat[]{\includegraphics[width=0.5\columnwidth]{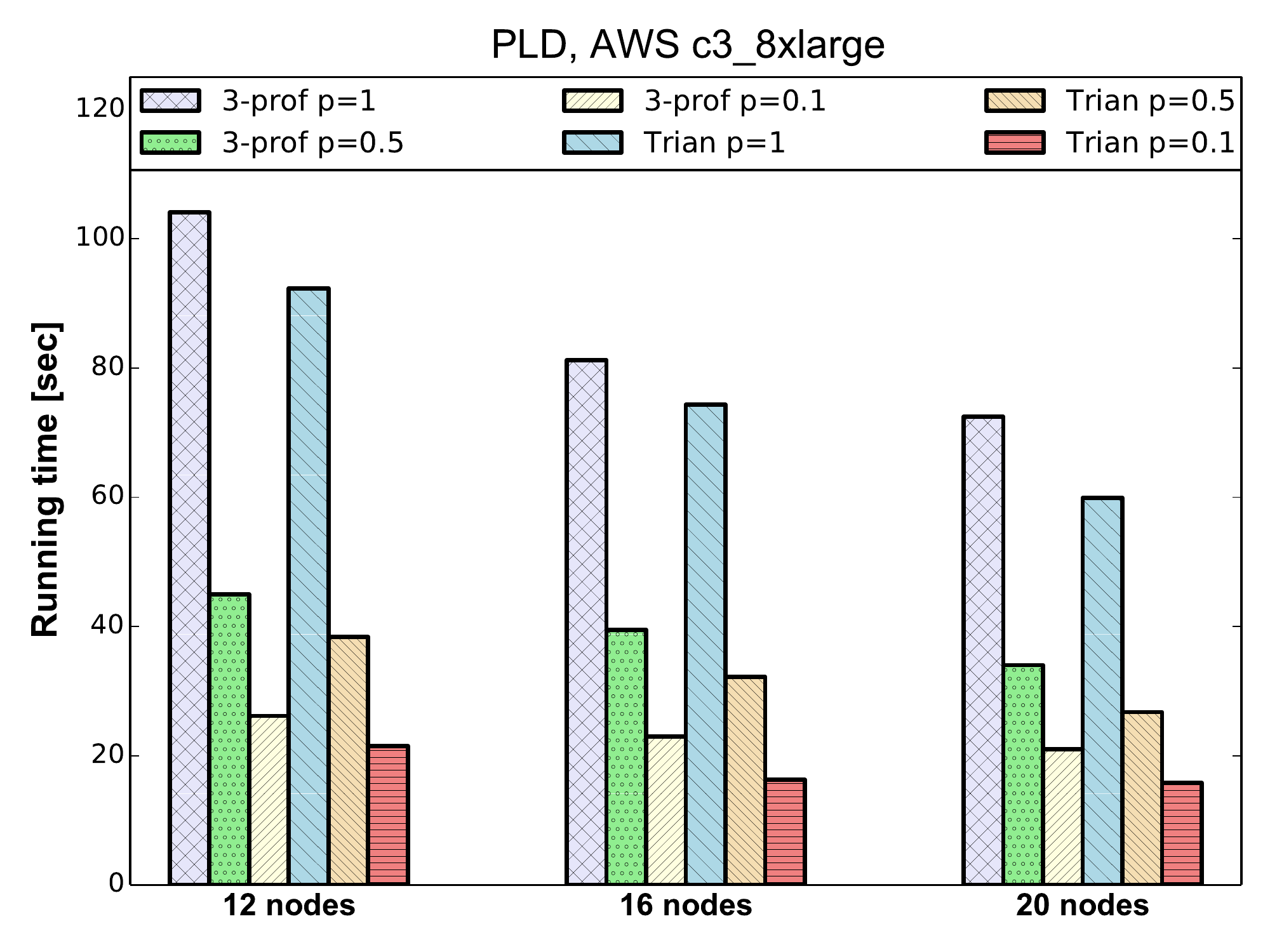}} \\ 
	\subfloat[]{\includegraphics[width=0.5\columnwidth]{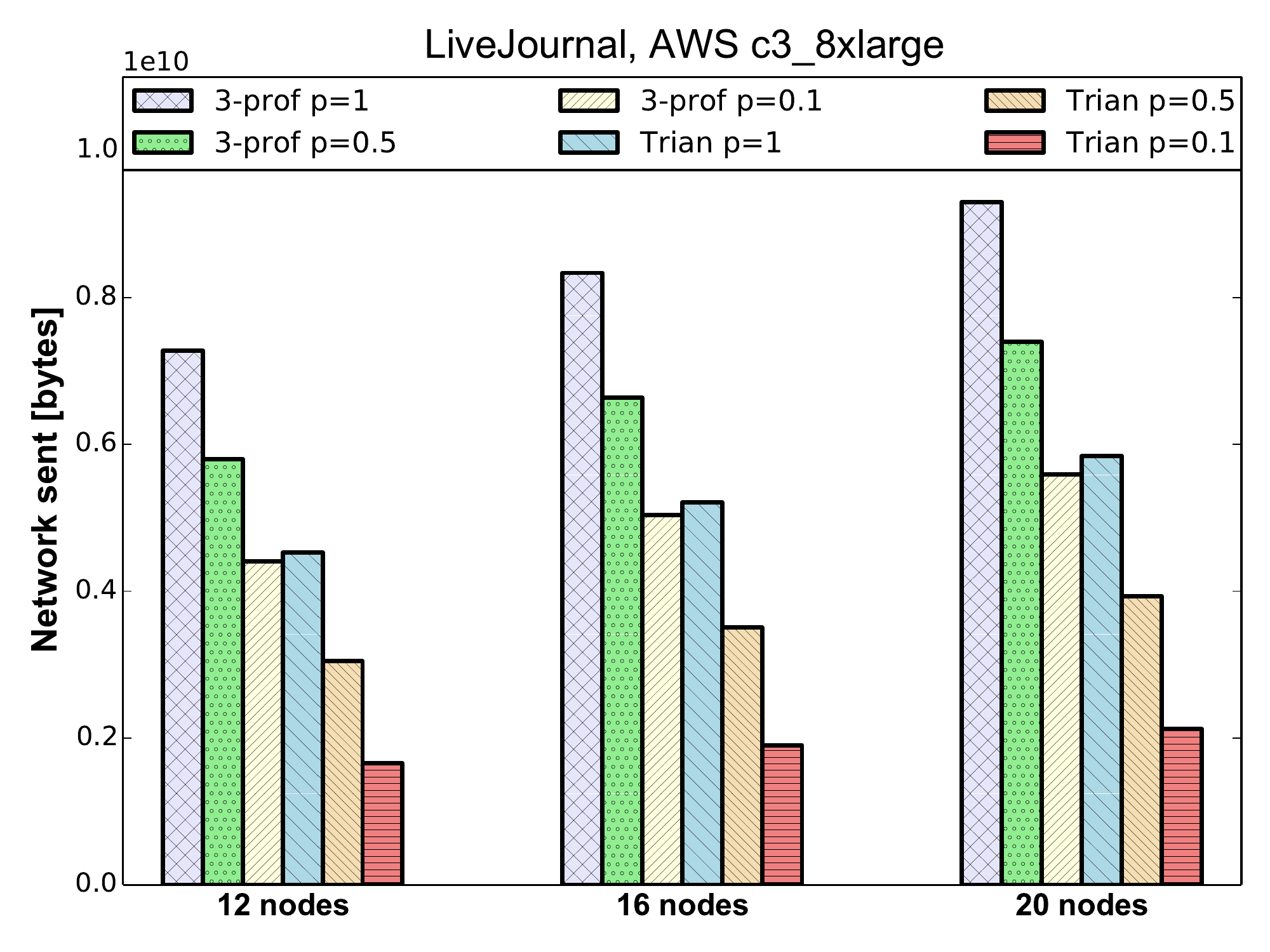}} 
	\subfloat[]{\includegraphics[width=0.5\columnwidth]{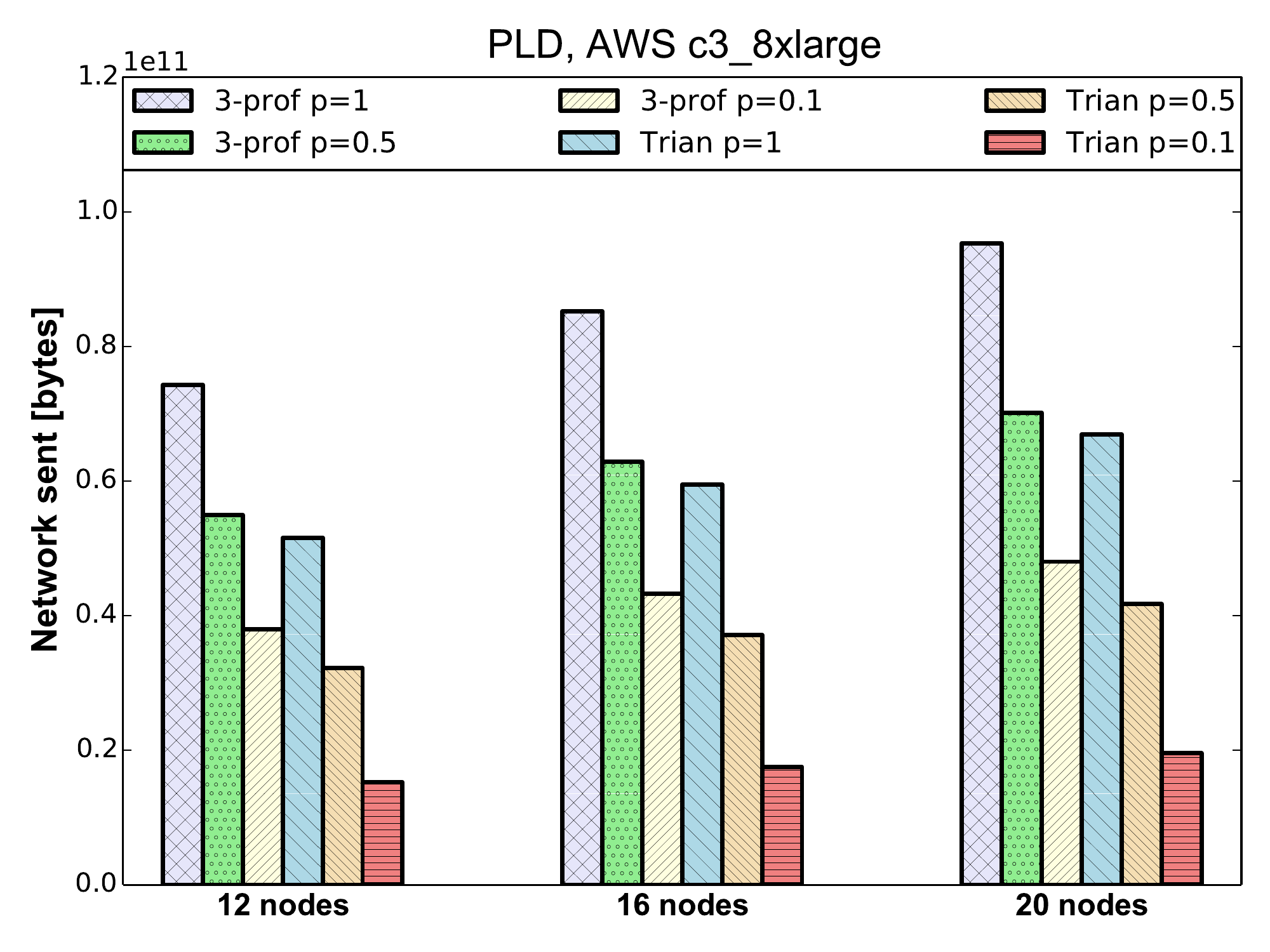}} \\
    \caption{\small AWS \texttt{c3\_8xlarge} cluster. $\prof$ vs. $\trian$ algorithms for LiveJournal and PLD datasets (average of $3$ runs). $\prof$ achieves comparable performance to triangle counting.
    (a,b) -- Running time for various numbers of nodes (machines) and various sampling probabilities $p$.
    (c,d) -- Network bytes sent by the algorithms for various numbers of nodes and various sampling probabilities $p$.
    }
    \label{fig:awsc3}
\end{figure}

\begin{figure}[h]
	\centering
	\subfloat[]{\includegraphics[width=0.5\columnwidth]{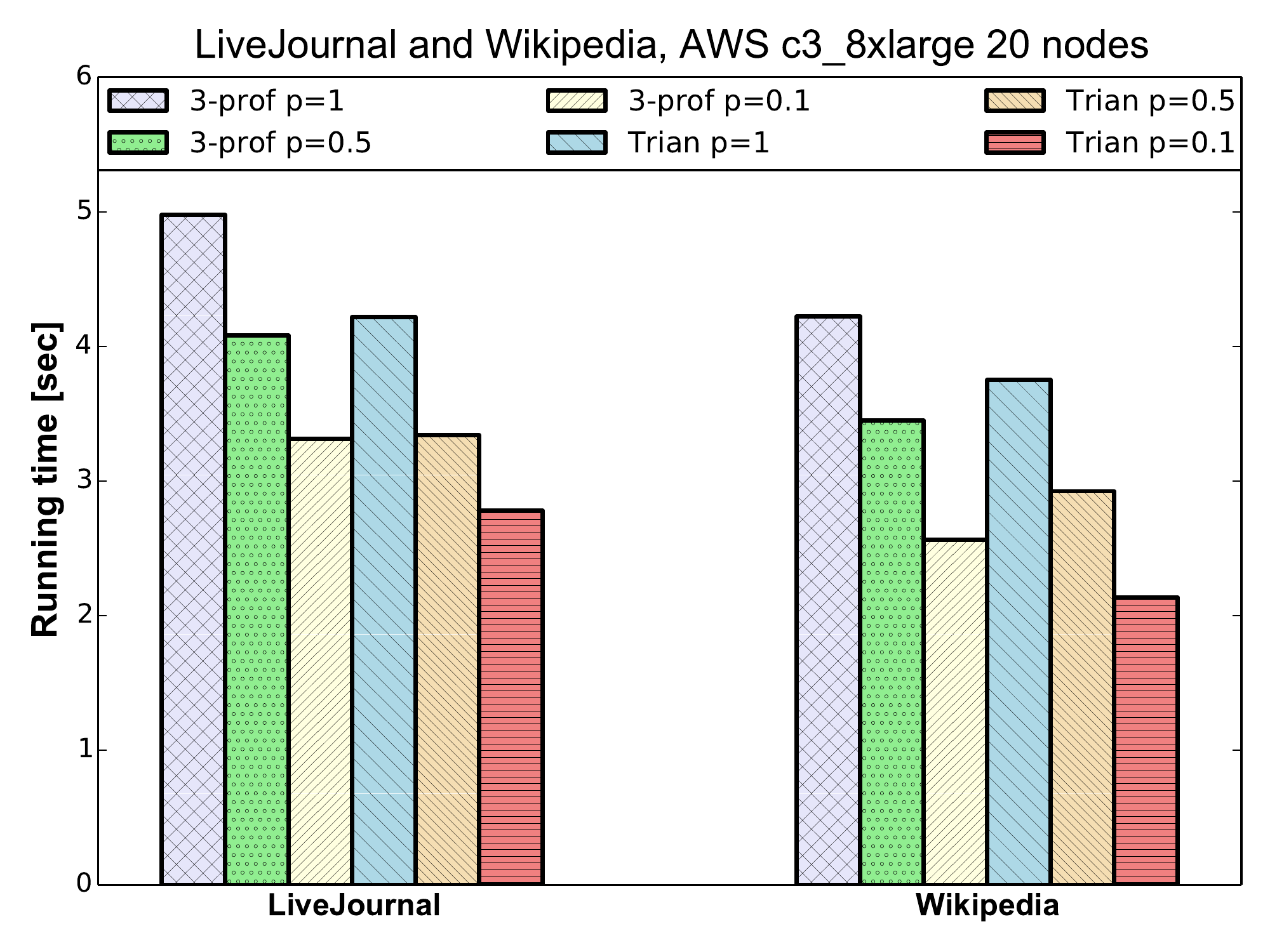}}
	\subfloat[]{\includegraphics[width=0.5\columnwidth]{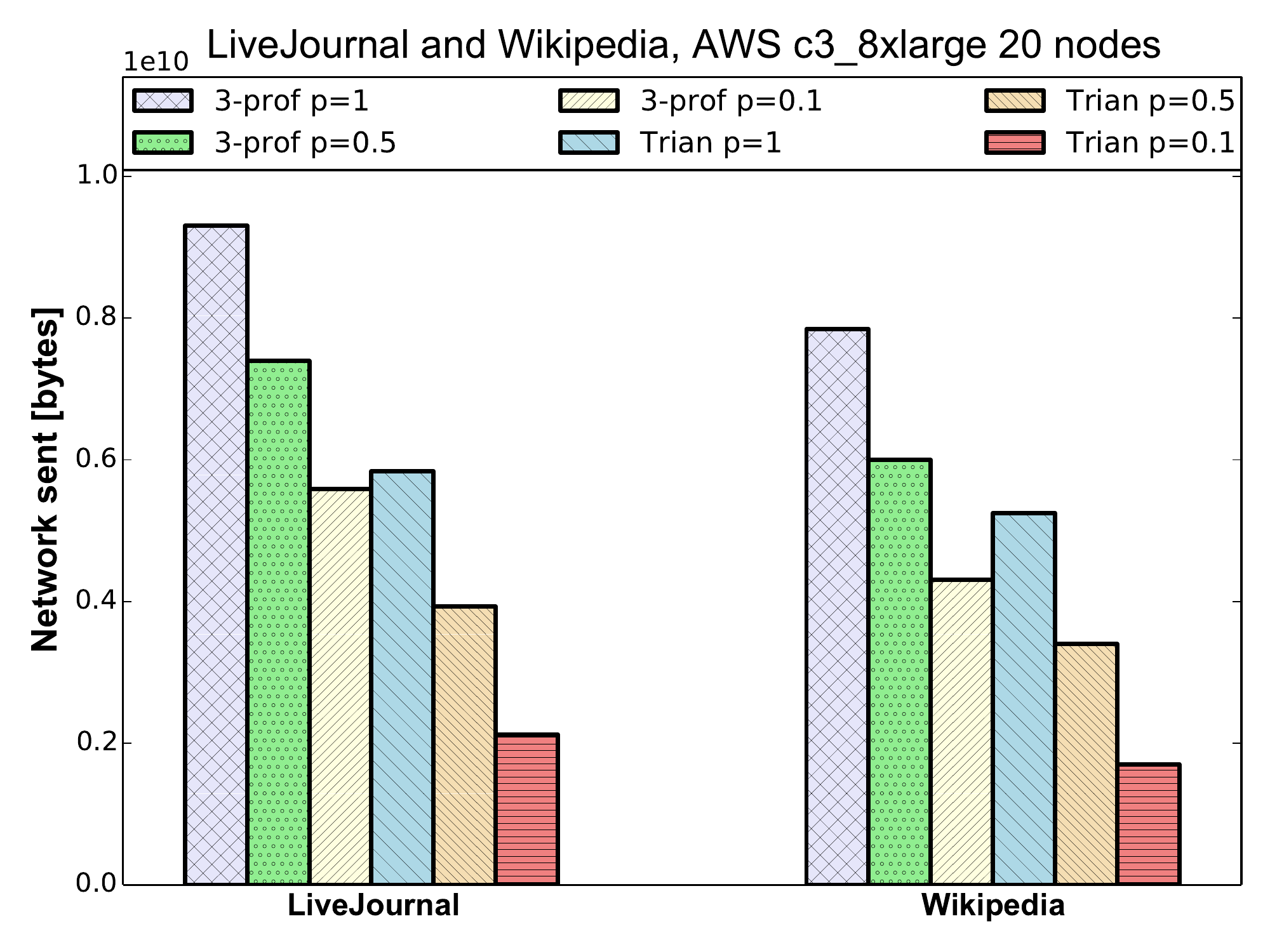}} \\
    \caption{\small AWS \texttt{c3\_8xlarge} cluster with 20 nodes. $\prof$ vs. $\trian$ results for LiveJournal and Wikipedia datasets (average of $3$ runs).
    (a) -- Running time for both graphs for various sampling probabilities $p$.
    (b) -- Network bytes sent by the algorithms for both graphs for various sampling probabilities $p$.
    }\label{fig:aws20node}
\end{figure}

\begin{figure}[h]
	\centering
	\subfloat[]{\includegraphics[width=0.5\columnwidth]{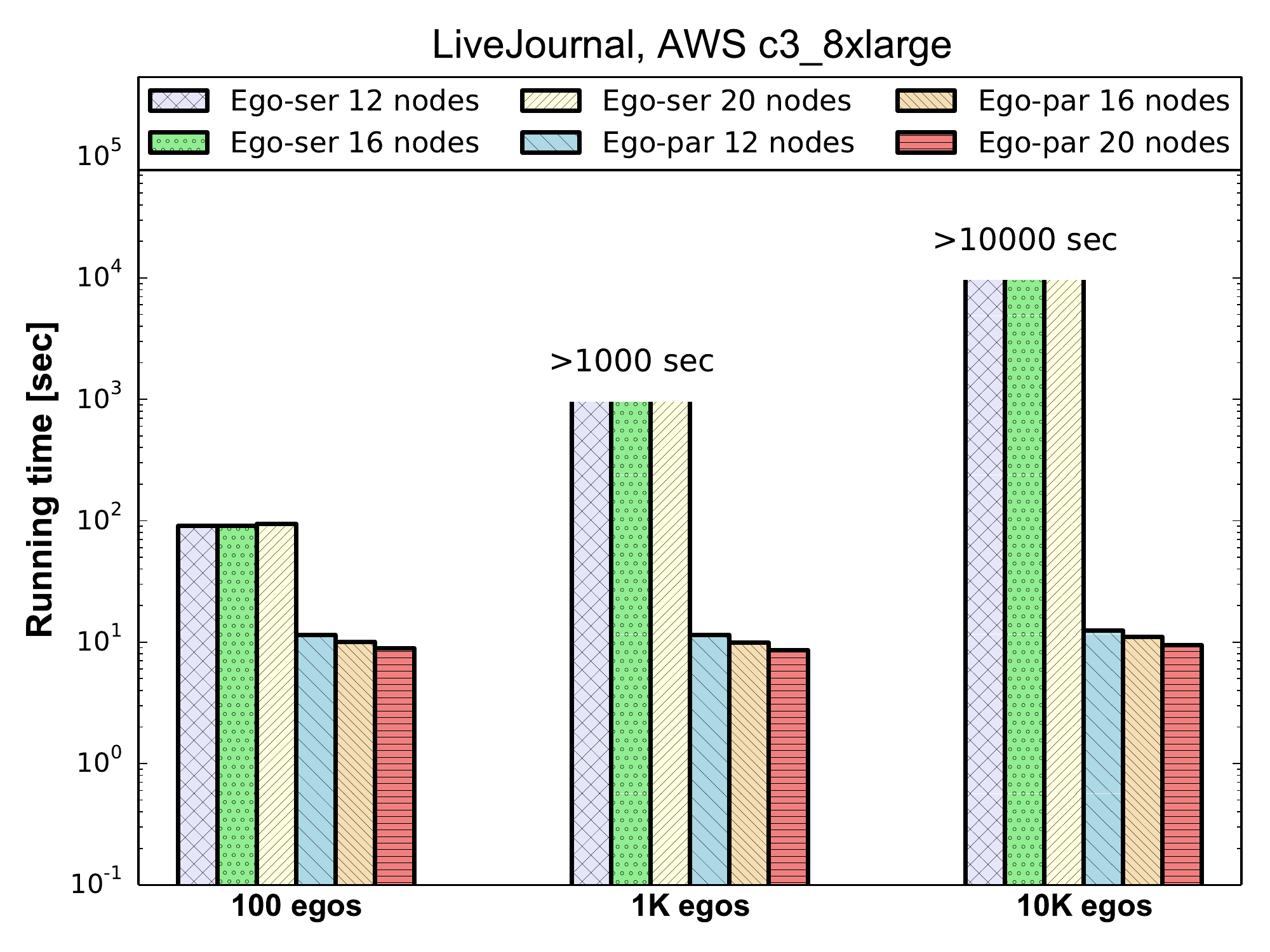}} 
	\subfloat[]{\includegraphics[width=0.5\columnwidth]{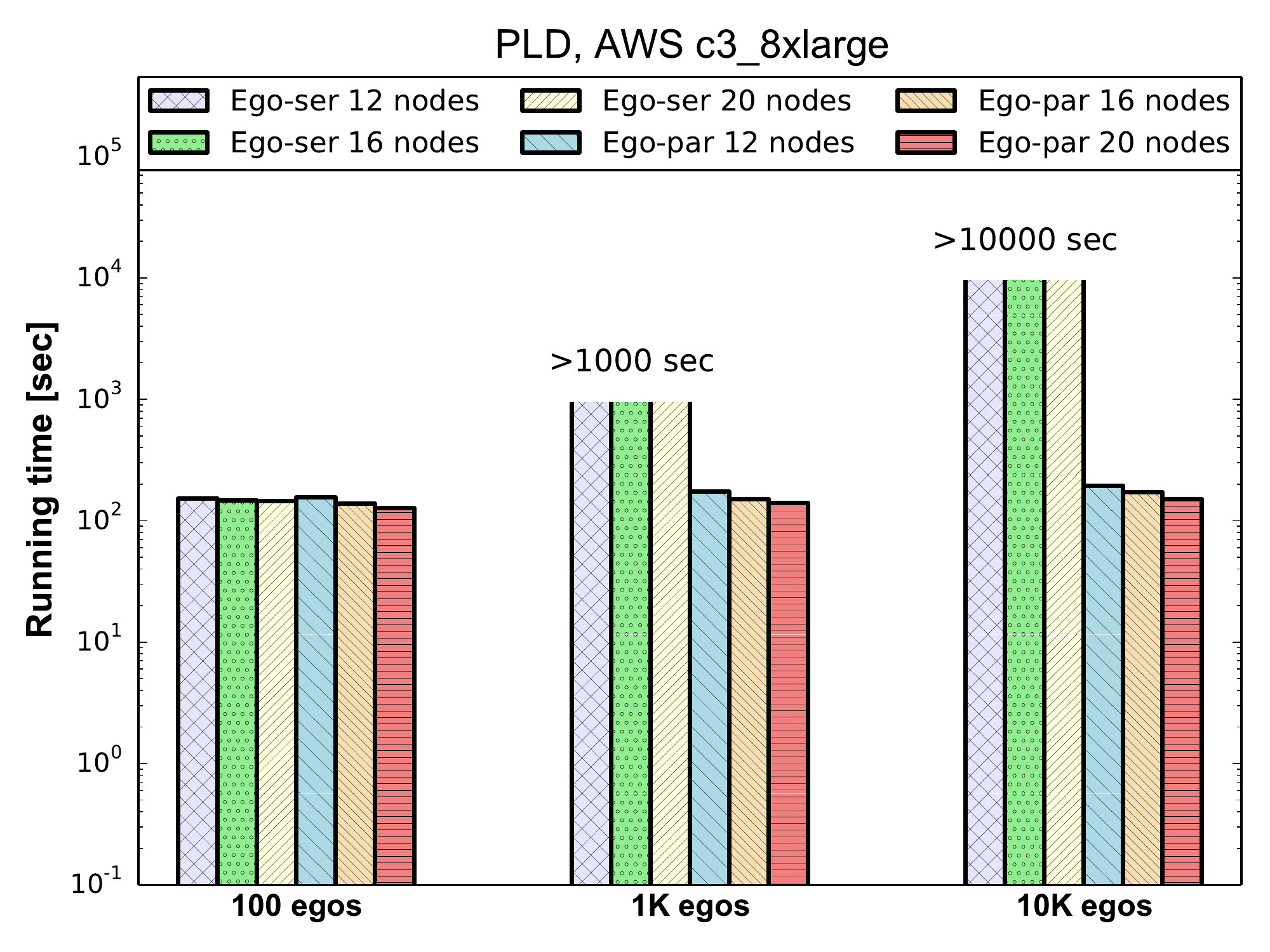}} \\ 
    \caption{\small AWS \texttt{c3\_8xlarge} cluster. $\egopar$ vs. $\egoser$  results for LiveJournal and PLD datasets (average of $5$ runs). Running time of $\egopar$ scales well with the number of ego centers, while $\egoser$ scales linearly.
    }
    \label{fig:awsc3ego}
\end{figure}

\noindent \textbf{Ego 3-profiles:}
The next set of experiments evaluates the performance of our $\egopar$ algorithm for counting ego $3$-profiles. We show the performance of $\egopar$ for various graphs and systems and also compare it to a naive serial algorithm $\egoser$. 
Let us start with the AWS system with (c3.8xlarge machines). In Figure \ref{fig:awsc3ego} we see the running time of $\egoser$ and $\egopar$ on the LiveJournal graph. The task was to find ego $3$-profiles of 100, 1K, and 10K randomly selected nodes. Since the running time depends on the size and structure of each induced subgraph, $\egoser$ and $\egopar$ operated on the same list of ego vertices. While for $100$ random vertices $\egoser$ performed well (and even achieved the same running time as $\egopar$ for the PLD graph), its performance drastically degraded for a larger number of vertices. This is due to its iterative nature -- it finds ego $3$-profiles of the vertices one at a time and is not scalable. 
Note that the open bars mean that this experiment was not finished. The numbers above them are extrapolations, which are reasonable due to the serial design of the $\egoser$. 

On the contrary, the $\egopar$ algorithm scales extremely well and computes ego $3$-profiles for 100, 1K, and 10K vertices almost in the same time. In Figure \ref{fig:awsc3egob} (a), we can see that as the number of nodes (i.e., machines) increases, running time of $\egopar$ decreases since its parallel design allows it to use additional computational resources. However, $\egoser$ cannot benefit from more resources and its running time even increases when more machines are used. The increase in running time of $\egoser$ is due to the increase in network usage when using more machines (see Figure \ref{fig:awsc3egob} (b)). The network usage of $\egopar$ also increases, but this algorithm compensates by leveraging additional computational power. 
In Figure \ref{fig:awsc3egoc}, we can see that $\egopar$ performs well even when finding ego $3$-profiles for all the LiveJournal vertices (4.8M vertices). 

Finally in Figure \ref{fig:asterixRuntimes} (b) and (c), we can see the comparison of $\egopar$ and $\egoser$ on the PLD and the DBLP graphs on the Asterix machine. 
For both graphs, we see a very good scaling of $\egopar$, while the running time of $\egoser$ scales linearly with the size of the ego vertices list.

\begin{figure*}[ht]
	\centering
	\subfloat[]{\includegraphics[width=0.67\columnwidth]{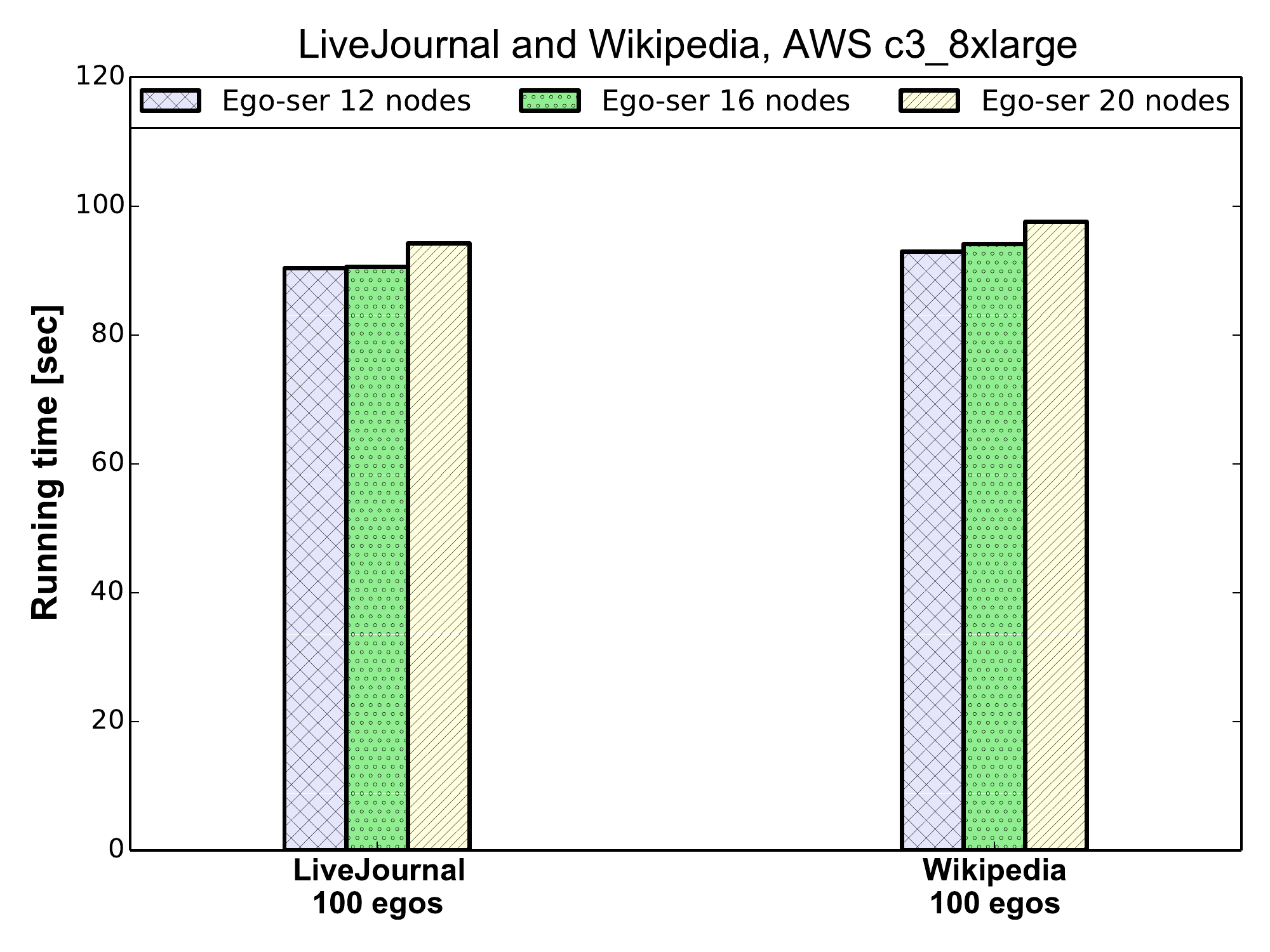}}  $\;$
		\subfloat[]{\includegraphics[width=0.67\columnwidth]{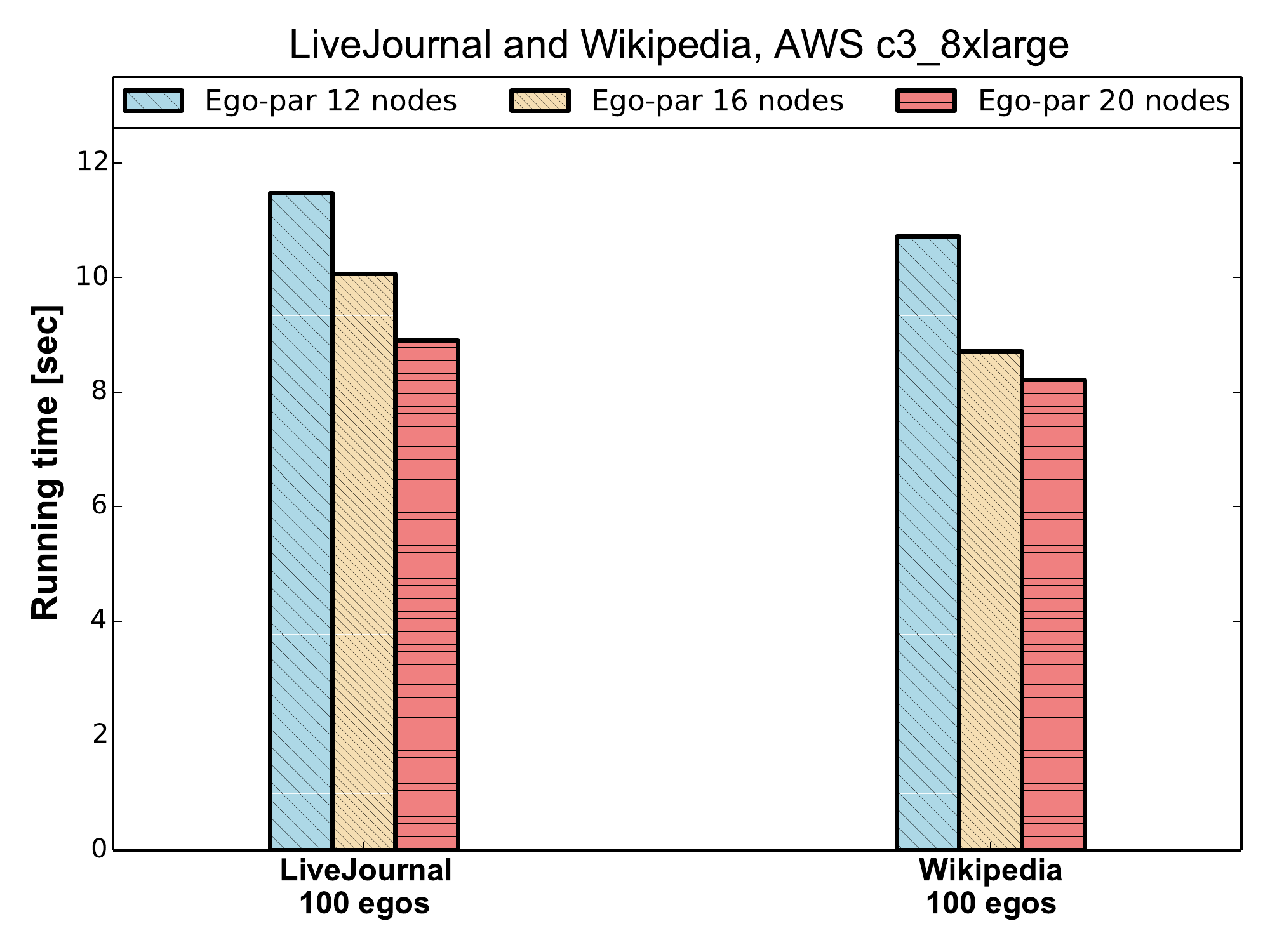}} $\;$
	\subfloat[]{\includegraphics[width=0.67\columnwidth]{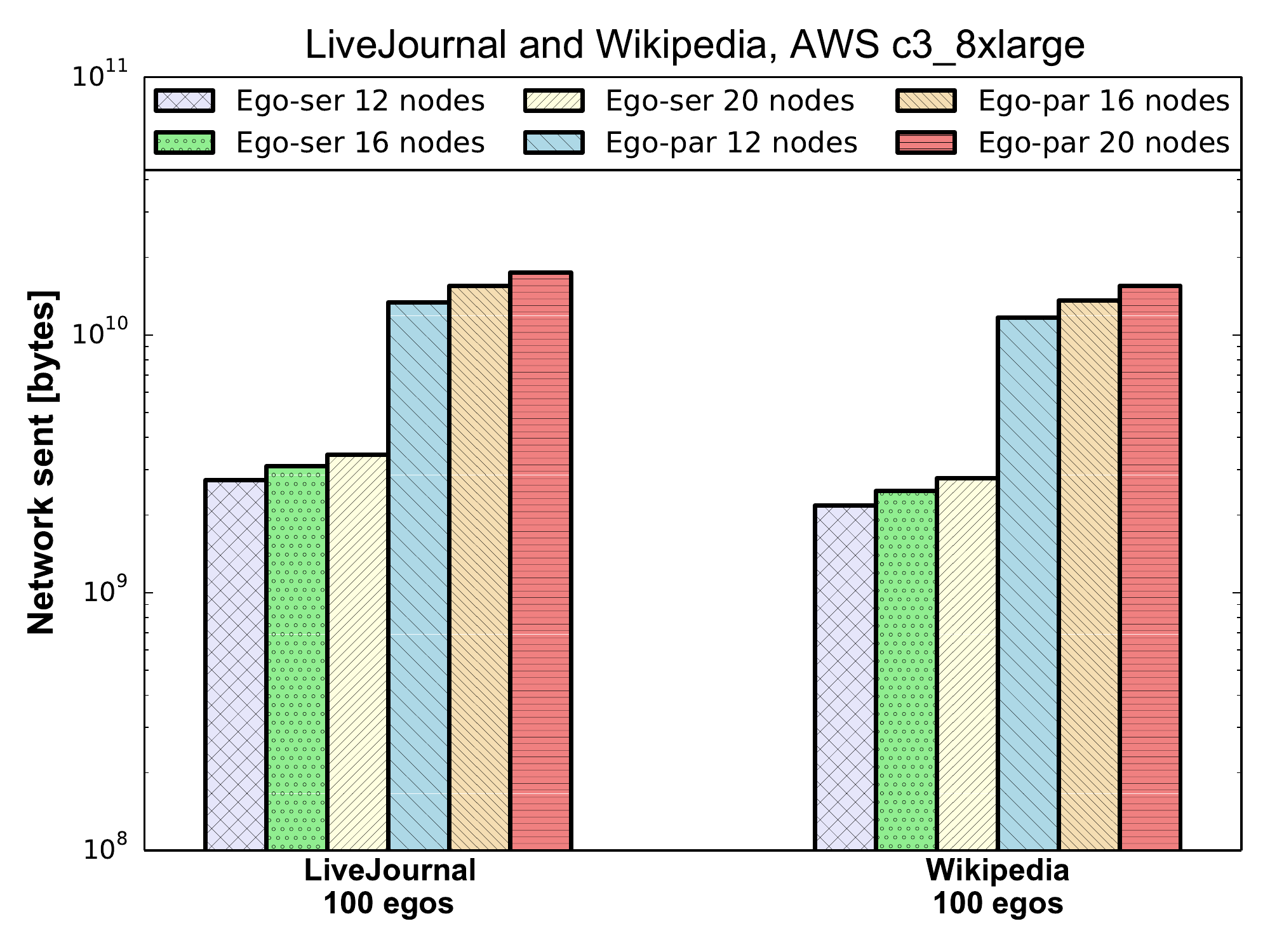}} 
    \caption{\small AWS \texttt{c3\_8xlarge} cluster. $\egopar$ vs. $\egoser$ results for LiveJournal and Wikipedia datasets (average of $5$ runs). 
Running time of $\egopar$ decreases with the number of machines due to its parallel design. Running time of $\egoser$ does not decrease with the number of machines due to its iterative nature. Network usage increases for both algorithms with the number of machines.
    }
    \label{fig:awsc3egob}
\end{figure*}

\begin{figure}[ht]
	\centering
	\subfloat[]{\includegraphics[width=0.5\columnwidth]{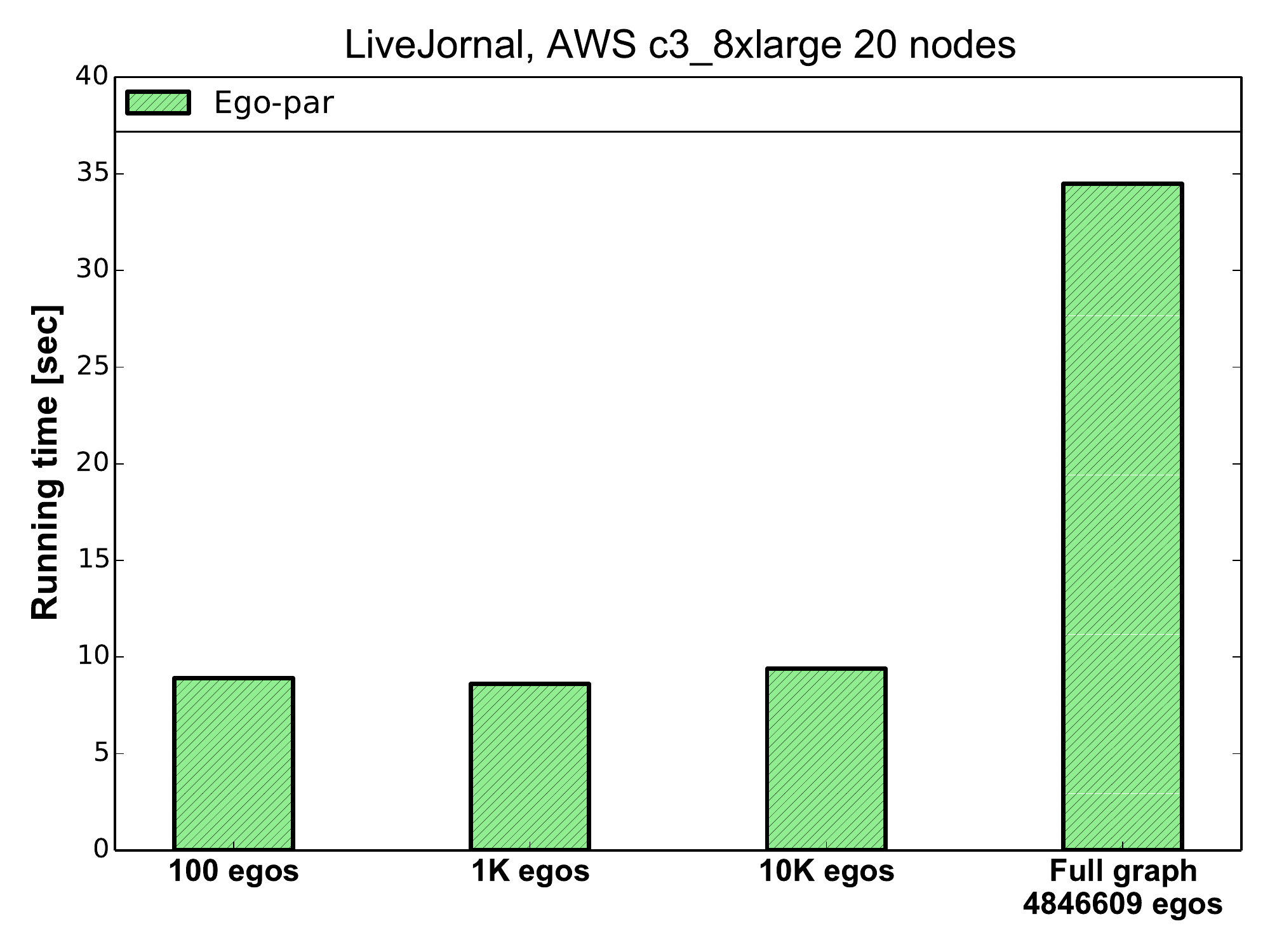}} 
	\subfloat[]{\includegraphics[width=0.5\columnwidth]{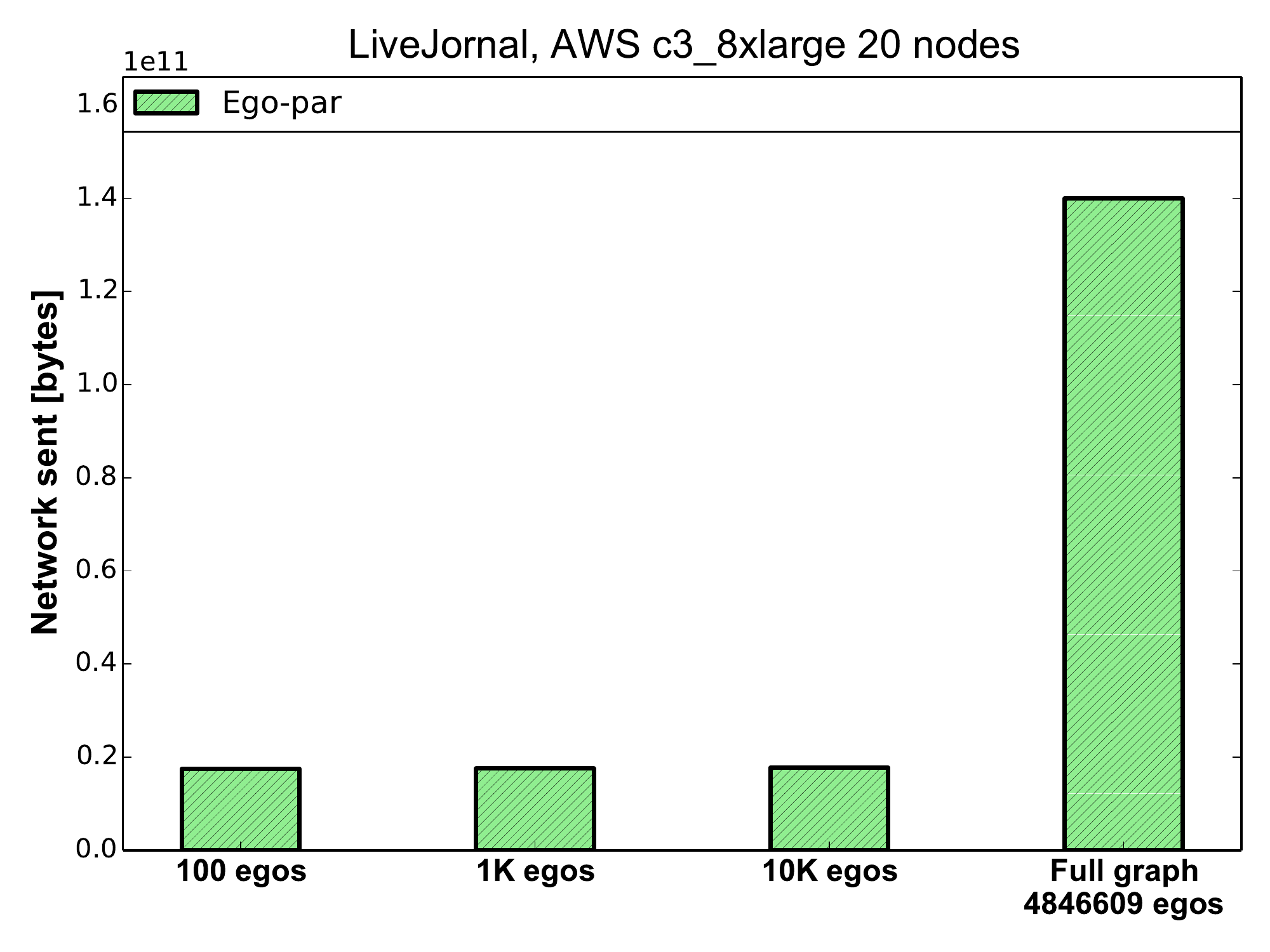}} \\
    \caption{\small AWS \texttt{c3\_8xlarge} cluster with 20 nodes. $\egopar$ results for LiveJournal dataset (average of $5$ runs). The algorithm scales well for various number of ego centers and even full ego centers list.
    (a) -- Running time. (b) -- Network bytes sent by the algorithm.
    }
    \label{fig:awsc3egoc}
\end{figure}

\begin{figure*}[ht]
	\centering
	\subfloat[]{\includegraphics[width=0.67\columnwidth]{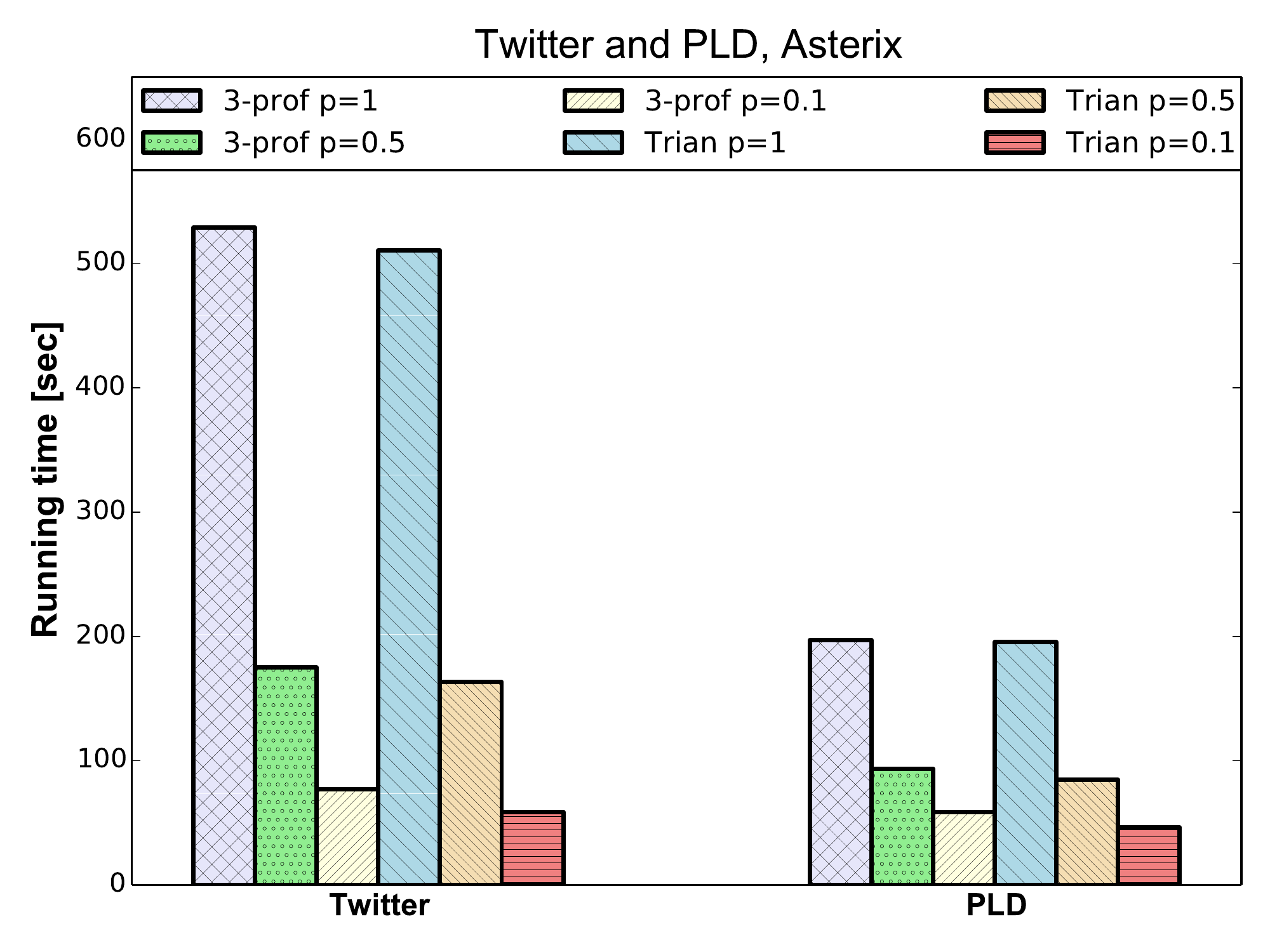}} $\;$
	\subfloat[]{\includegraphics[width=0.67\columnwidth]{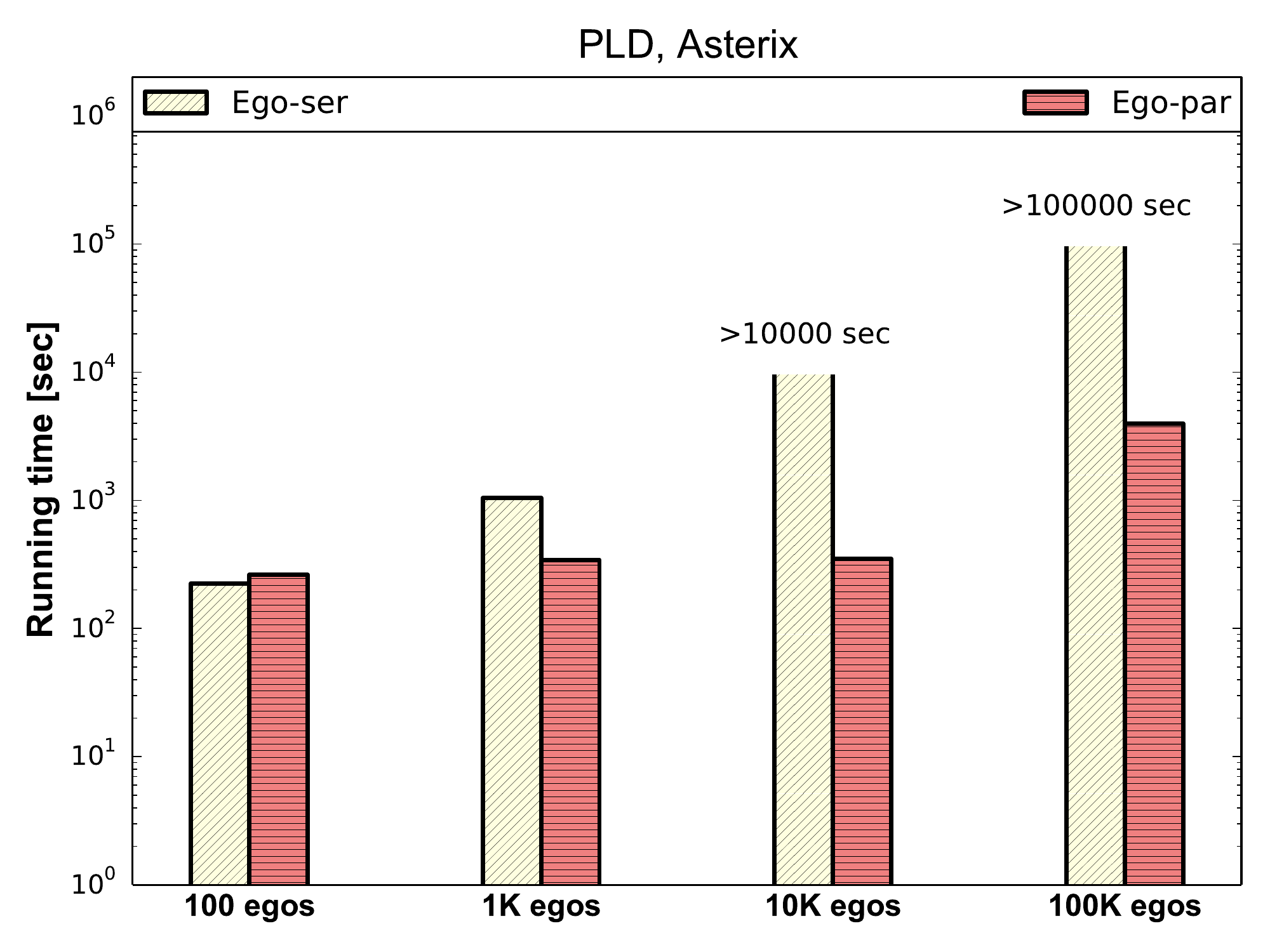}} $\;$
	\subfloat[]{\includegraphics[width=0.67\columnwidth]{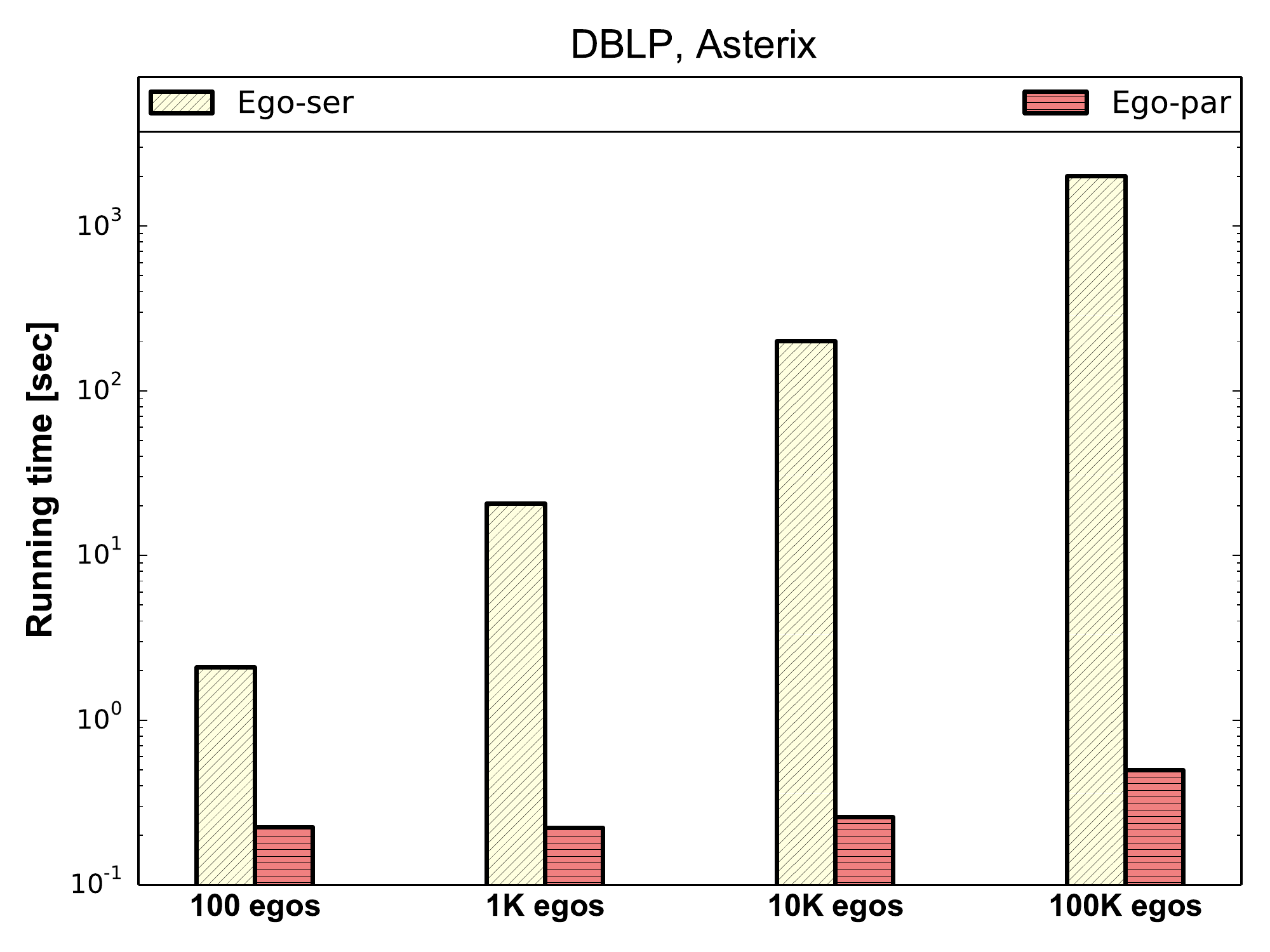}} 
    \caption{\small Asterix machine. Results for Twitter, PLD, and DBLP datasets. 
    (a) -- Running time of $\prof$ vs. $\trian$ for various sampling probabilities $p$.
    (b,c) -- Running time of $\egopar$ vs. $\egoser$ for various number of ego centers. Results are averaged over $3$, and $3$, and $10$ runs, respectively. 
    }\label{fig:asterixRuntimes}
\end{figure*}

\begin{figure*}[ht]
	\centering
	\subfloat[]{\includegraphics[width=0.67\columnwidth]{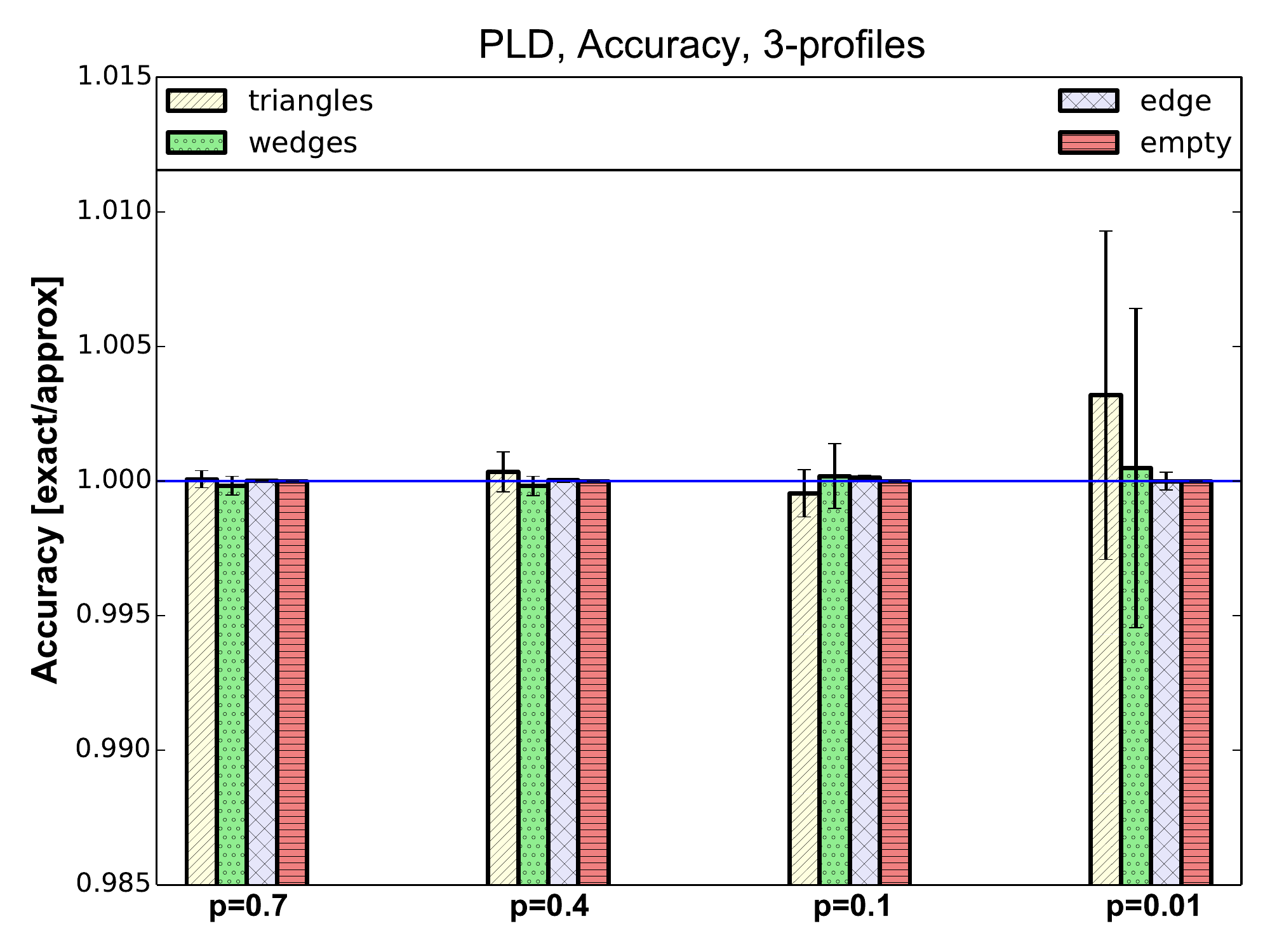}} $\;$
	\subfloat[]{\includegraphics[width=0.67\columnwidth]{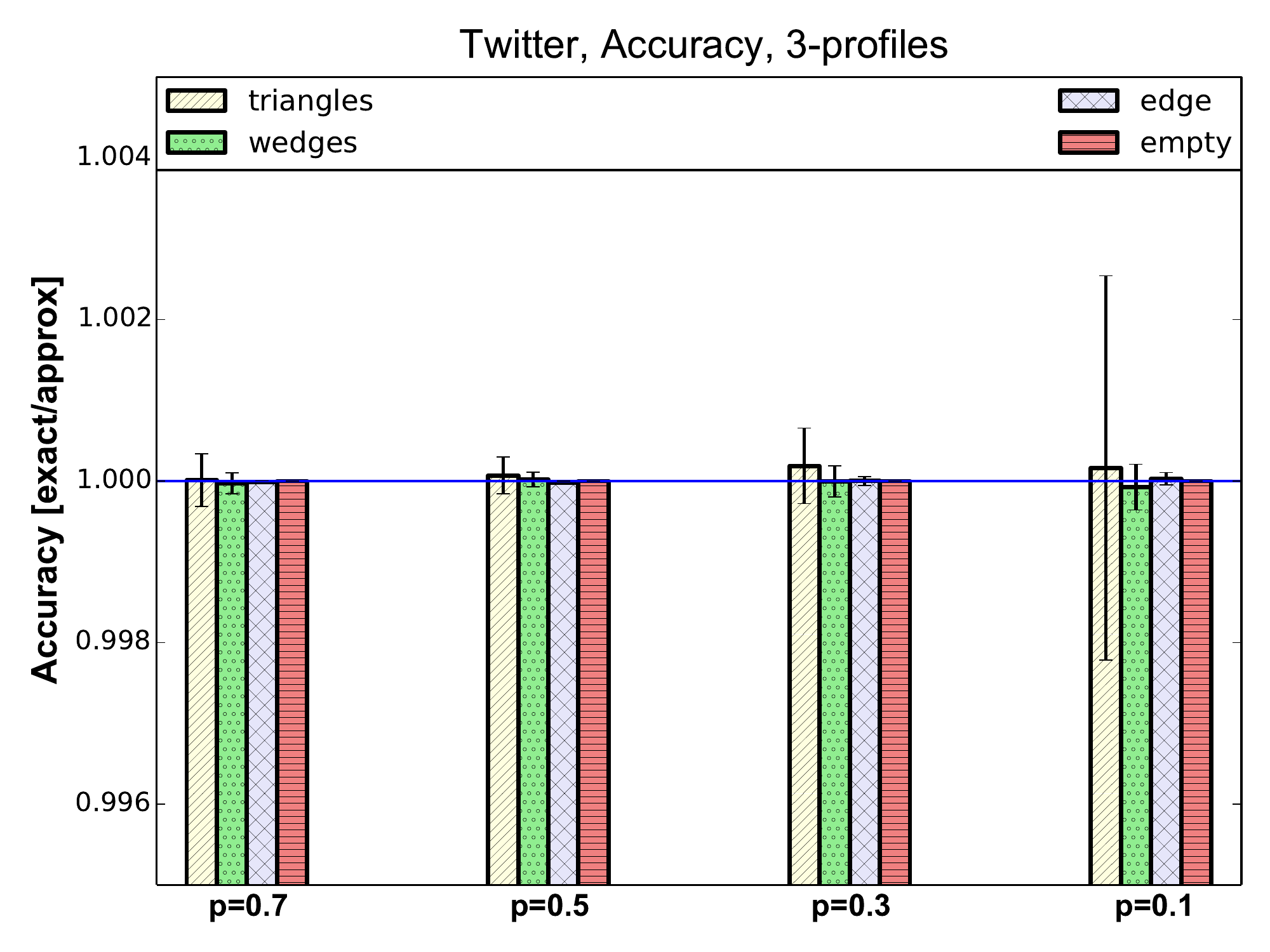}} $\;$
	\subfloat[]{\includegraphics[width=0.67\columnwidth]{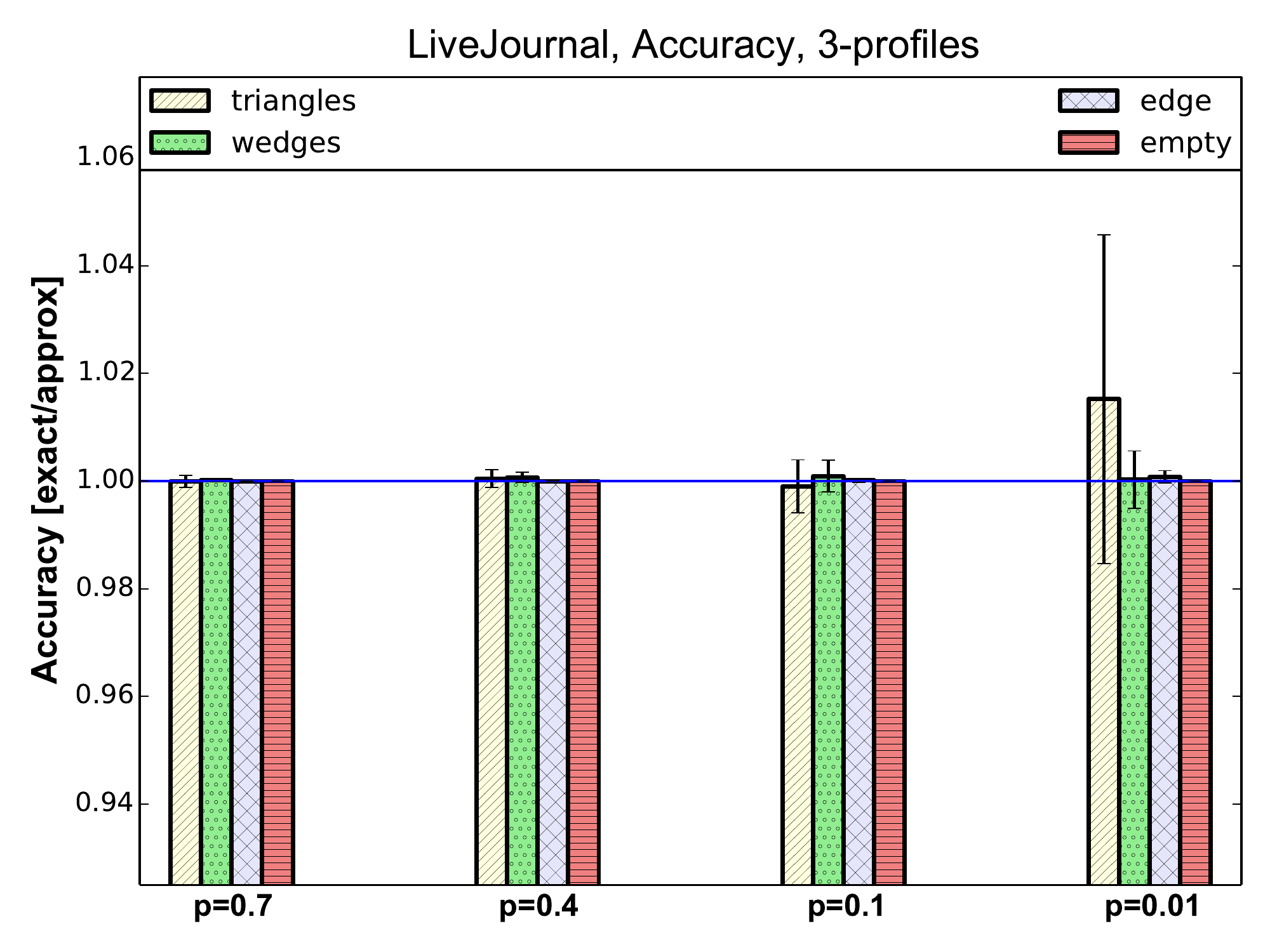}} 
    \caption{\small Global $3$-profiles accuracy achieved by $\prof$ algorithm for various graphs and for each profile count. Results are averaged over $5$, and $5$, and $10$ iterations, respectively. Error bars indicate $1$ standard deviation. The metric is a ratio between the exact profile count (when $p=1$) and the given output for $p<1$. All results are very close to the optimum value of $1$.
    }
    \label{fig:3profAcc}
\end{figure*}

\section{Conclusions}
In summary, we have reduced several $3$-profile problems to triangle and $4$-clique finding in a graph engine framework. Our concentration theorem and experimental results confirm that local $3$-profile estimation via sub-sampling is comparable in runtime and accuracy to local triangle counting.

This paper offers several directions for future work. First, both the local $3$-profile and the ego $3$-profile can be used as features to classify vertices in social or bioinformatic networks. Additionally, we hope to extend our theory and algorithmic framework to larger subgraphs, as well as special classes of input graphs. Our edge sampling Markov chain and unbiased estimators should easily extend to $k~>~3$. Equations in \eqref{eq:trisq} are useful to count local or global $4$-profiles in a centralized setting, as shown recently in \cite{Kowaluk2013eqs,Williams2014mod}. 
Tractable distributed algorithms for $k~>~3$ using similar edge pivot equations remain as future work.
Our observed dependence on $4$-clique count suggests that an improved graph engine-based clique counting subroutine will improve the parallel algorithm's performance.

\newpage
\bibliographystyle{abbrv}
\bibliography{refs}

\begin{thebibliography}{10}

\bibitem{Ahmed2014}
N.~K. Ahmed, N.~Duffield, J.~Neville, and R.~Kompella.
\newblock Graph sample and hold: {A} framework for big-graph analytics.
\newblock In {\em KDD}, 2014.

\bibitem{Alon1997}
N.~Alon, R.~Yuster, and U.~Zwick.
\newblock Finding and counting given length cycles.
\newblock {\em Algorithmica}, 17(3):209--223, Mar. 1997.

\bibitem{aws}
{\relax Amazon web services}.
\newblock http://aws.amazon.com, 2015.

\bibitem{becchetti08}
L.~Becchetti, P.~Boldi, C.~Castillo, and A.~Gionis.
\newblock Efficient semi-streaming algorithms for local triangle counting in
  massive graphs.
\newblock In {\em KDD}, 2008.

\bibitem{Bhuiyan2012}
M.~A. Bhuiyan, M.~Rahman, M.~Rahman, and M.~{Al Hasan}.
\newblock Guise: {U}niform sampling of graphlets for large graph analysis.
\newblock {\em IEEE 12th International Conference on Data Mining}, pages
  91--100, Dec. 2012.

\bibitem{Borgs2006}
C.~Borgs, J.~Chayes, and K.~Vesztergombi.
\newblock Counting graph homomorphisms.
\newblock {\em Topics in Discrete Mathematics}, pages 315--371, 2006.

\bibitem{Chu2011}
S.~Chu and J.~Cheng.
\newblock Triangle listing in massive networks and its applications.
\newblock In {\em Proc. 17th ACM SIGKDD}, page 672, New York, New York, USA,
  2011.

\bibitem{wikigraph}
T.~A. Davis and Y.~Hu.
\newblock The {U}niversity of {F}lorida sparse matrix collection.
\newblock {\em ACM Transactions on Mathematical Software}, 38(1):1--25, 2011.

\bibitem{powergraphGAS2012}
J.~E. Gonzalez, Y.~Low, H.~Gu, D.~Bickson, and C.~Guestrin.
\newblock {PowerGraph}: {D}istributed graph-parallel computation on natural
  graphs.
\newblock In {\em 10th USENIX Symposium on Operating Systems Design and
  Implementation (OSDI)}, pages 17--30, 2012.

\bibitem{Han2013}
W.~Han and J.~Lee.
\newblock Turbo$_{ISO}$: {T}owards ultrafast and robust subgraph isomorphism
  search in large graph databases.
\newblock In {\em SIGMOD}, pages 337--348, 2013.

\bibitem{Hormozdiari2007}
F.~Hormozdiari, P.~Berenbrink, N.~Przulj, and S.~C. Sahinalp.
\newblock Not all scale-free networks are born equal: {T}he role of the seed
  graph in {PPI} network evolution.
\newblock {\em PLoS computational biology}, 3(7):e118, July 2007.

\bibitem{Hocevar2014bio}
T.~Ho\v{c}evar and J.~Dem\v{s}ar.
\newblock A combinatorial approach to graphlet counting.
\newblock {\em Bioinformatics}, 30(4):559--65, Feb. 2014.

\bibitem{Jha2012Birthday}
M.~Jha, C.~Seshadhri, and A.~Pinar.
\newblock A space efficient streaming algorithm for triangle counting using the
  birthday paradox.
\newblock In {\em KDD}, pages 589--597, 2013.

\bibitem{Jha2014}
M.~Jha, C.~Seshadhri, and A.~Pinar.
\newblock Path sampling: {A} fast and provable method for estimating 4-vertex
  subgraph counts.
\newblock 2014.

\bibitem{KimVu2000concentration}
J.~H. Kim and V.~H. Vu.
\newblock Concentration of multivariate polynomials and its applications.
\newblock {\em Combinatorica}, 20(3):417--434, 2000.

\bibitem{Kloks2000eqs}
T.~Kloks, D.~Kratsch, and H.~M\"{u}ller.
\newblock Finding and counting small induced subgraphs efficiently.
\newblock {\em Information Processing Letters}, 74(3-4):115--121, May 2000.

\bibitem{Kowaluk2013eqs}
M.~Kowaluk, A.~Lingas, and E.-M. Lundell.
\newblock Counting and detecting small subgraphs via equations.
\newblock {\em SIAM Journal of Discrete Mathematics}, 27(2):892--909, 2013.

\bibitem{twitterKwak10www}
H.~Kwak, C.~Lee, H.~Park, and S.~Moon.
\newblock {W}hat is {T}witter, a social network or a news media?
\newblock In {\em Proc. 19th International World Wide Web Conference}, pages
  591--600, New York, NY, USA, 2010. ACM.

\bibitem{Lee2012}
J.~Lee, W.-S. Han, R.~Kasperovics, and J.-H. Lee.
\newblock An in-depth comparison of subgraph isomorphism algorithms in graph
  databases.
\newblock {\em Proc. VLDB Endowment}, 6(2):133--144, Dec. 2012.

\bibitem{snapnets}
J.~Leskovec and A.~Krevl.
\newblock {SNAP Datasets}: {Stanford} large network dataset collection.
\newblock \url{http://snap.stanford.edu/data}, June 2014.

\bibitem{lovasz2012large}
L.~Lov{\'a}sz.
\newblock {\em Large Networks and Graph Limits}, volume~60.
\newblock American Mathematical Soc., 2012.

\bibitem{Marcus2012rage}
D.~Marcus and Y.~Shavitt.
\newblock {RAGE} - {A} rapid graphlet enumerator for large networks.
\newblock {\em Computer Networks}, 56(2):810--819, Feb. 2012.

\bibitem{pldgraph}
R.~Meusel, S.~Vigna, O.~Lehmberg, and C.~Bizer.
\newblock Graph structure in the web -- revisited.
\newblock In {\em Proc. 23rd International World Wide Web Conference, Web
  Science Track}. ACM, 2014.

\bibitem{Milenkovik2008}
T.~Milenkovik and N.~Przulj.
\newblock Uncovering biological network function via graphlet degree
  signatures.
\newblock {\em Cancer Informatics}, 6:257--273, 2008.

\bibitem{OCallaghan2012}
D.~O'Callaghan, M.~Harrigan, J.~Carthy, and P.~Cunningham.
\newblock Identifying discriminating network motifs in youtube spam.
\newblock Feb. 2012.

\bibitem{Pagh2012}
R.~Pagh and C.~E. Tsourakakis.
\newblock Colorful triangle counting and a {M}ap{R}educe implementation.
\newblock {\em Information Processing Letters}, 112(7):277--281, Mar. 2012.

\bibitem{przPPIorig}
N.~Przulj.
\newblock Biological network comparison using graphlet degree distribution.
\newblock {\em Bioinformatics}, 23(2):177--183, 2007.

\bibitem{Ribeiro2010}
P.~Ribeiro, F.~Silva, and L.~Lopes.
\newblock Efficient parallel subgraph counting using g-tries.
\newblock In {\em IEEE International Conference on Cluster Computing}, pages
  217--226. Ieee, Sept. 2010.

\bibitem{Saltz2014}
M.~Saltz, A.~Jain, A.~Kothari, A.~Fard, J.~A. Miller, and L.~Ramaswamy.
\newblock Dual{I}so: {A}n algorithm for subgraph pattern matching on very large
  labeled graphs.
\newblock {\em IEEE International Congress on Big Data}, pages 498--505, June
  2014.

\bibitem{Satish2014}
N.~Satish, N.~Sundaram, M.~A. Patwary, J.~Seo, J.~Park, M.~A. Hassaan,
  S.~Sengupta, Z.~Yin, and P.~Dubey.
\newblock Navigating the maze of graph analytics frameworks using massive graph
  datasets.
\newblock In {\em SIGMOD}, pages 979--990, 2014.

\bibitem{Schank2007}
T.~Schank.
\newblock {\em Algorithmic Aspects of Triangle-Based Network Analysis}.
\newblock PhD thesis, 2007.

\bibitem{seshadri2012wedge}
C.~Seshadhri, A.~Pinar, and T.~G. Kolda.
\newblock Triadic measures on graphs: {T}he power of wedge sampling.
\newblock In {\em Proc. SIAM Conference on Data Mining}, pages 10--18, 2013.

\bibitem{Seshadhri2013wedge}
C.~Seshadhri, A.~Pinar, and T.~G. Kolda.
\newblock Wedge sampling for computing clustering coefficients and triangle
  counts on large graphs.
\newblock {\em Statistical Analysis and Data Mining}, 7(4):294--307, 2014.

\bibitem{Shervashidze2009Learning}
N.~Shervashidze, K.~Mehlhorn, and T.~H. Petri.
\newblock Efficient graphlet kernels for large graph comparison.
\newblock In {\em Proc. 20th International Conference on Artificial
  Intelligence and Statistics}, pages 488--495, 2009.

\bibitem{SuriReducer}
S.~Suri and S.~Vassilvitskii.
\newblock Counting triangles and the curse of the last reducer.
\newblock In {\em Proc. 20th International World Wide Web Conference}, page
  607, 2011.

\bibitem{Tsourakakis}
C.~E. Tsourakakis.
\newblock Fast counting of triangles in large real networks: {A}lgorithms and
  laws.
\newblock In {\em IEEE International Conference on Data Mining}, 2008.

\bibitem{TsourakakisSubsamp2009}
C.~E. Tsourakakis, U.~Kang, G.~L. Miller, and C.~Faloutsos.
\newblock Doulion: {C}ounting triangles in massive graphs with a coin.
\newblock In {\em SIGKDD}, 2009.

\bibitem{Tsourakakis2011sparsifier}
C.~E. Tsourakakis, M.~Kolountzakis, and G.~L. Miller.
\newblock Triangle sparsifiers.
\newblock {\em Journal of Graph Theory and Applications}, 15(6):703--726, 2011.

\bibitem{Ugander2013}
J.~Ugander, L.~Backstrom, M.~Park, and J.~Kleinberg.
\newblock Subgraph frequencies: {M}apping the empirical and extremal geography
  of large graph collections.
\newblock In {\em 22nd International World Wide Web Conference}, 2013.

\bibitem{Williams2014mod}
V.~V. Williams, J.~Wang, R.~Williams, and H.~Yu.
\newblock Finding four-node subgraphs in triangle time.
\newblock {\em SODA}, pages 1671--1680, 2014.

\bibitem{Yan2002}
X.~Yan and J.~Han.
\newblock {gSpan}: {G}raph-based substructure pattern mining.
\newblock In {\em International Conference on Data Mining}, 2002.

\end{thebibliography}

\newpage
\appendix
\section{ Proof of Theorem \ref{PROFILECONCENTRATION}}
Let $m$ be the total number of edges in the original graph $G$. If $e$ is an edge in the original graph $G$, let $t_e$ be the random indicator after sampling. $t_e=1$ if $e$ is sampled and $0$ otherwise. Let ${\cal H}_0, {\cal H}_1, {\cal H}_2,{\cal H}_3$ denote the set of distinct subgraphs of the kind $H_0,H_1,H_2$ and $H_3$ (anti-clique, edge, wedge and triangle) respectively. Let $A,\_(e),\Lambda(e,f)$ and $\Delta(e,f,g)$ denote an anti-clique with no edges, a $H_1$ with edge $e$ , a $H_2$ with two edges $e,f$ and a triangle with edges $e,f,g$ respectively in the original graph $G$. Our estimators \eqref{eq:estimatorStart}-\eqref{eq:estimatorEnd} are a function of $Y_i$'s and each $Y_i$ can be written as a polynomial of at most degree $3$ in all the variables $t_e$.
\begin{align}
Y_0 &= n_0+ \sum_{\_(e) \in {\cal H}_1} (1-t_{e})+ \sum_{\Lambda(e,f) \in {\cal H}_2} (1-t_e)(1-t_f) + \nonumber \\ \
  \hfill & \sum_{\Delta(e,f,g) \in {\cal H}_3} (1-t_e)(1-t_f)(1-t_g) \label{eq:Y0} \\ 
  Y_1 &= \sum_{\_(e) \in {\cal H}_1} t_{e}+ \sum_{\Lambda(e,f) \in {\cal H}_2} ((1-t_e)t_f +(1-t_f)t_e ) + \nonumber \\ 
 \hfill & \sum_{\Delta(e,f,g) \in {\cal H}_3} t_e(1-t_f)(1-t_g)+ \nonumber \\
 \hfill & \sum_{\Delta(e,f,g) \in {\cal H}_3} t_f(1-t_e)(1-t_g)+t_g(1-t_e)(1-t_f)   \label{eq:Y1}\\ 
 Y_2 &= \sum_{\Lambda(e,f) \in {\cal H}_2} t_e t_f + \nonumber \\
 \hfill & \sum_{\Delta(e,f,g) \in {\cal H}_3} (t_et_f)(1-t_g)+t_f t_g(1-t_e)+t_e(1-t_f)t_g \label{eq:Y2}  \\ 
 Y_3 &= \sum_{\Delta(e,f,g) \in {\cal H}_3} t_et_f t_g \label{eq:Y3} \\
 S_1 &=  \sum_{\_(e) \in {\cal H}_1} t_{e} \label{eq:S1} \\
 D_1 & = \sum_{\Lambda(e,f) \in {\cal H}_2} (t_e + t_f) \label{eq:D1} \\
 D_2 &= \sum_{\Lambda(e,f) \in {\cal H}_2} t_et_f \label{eq:D2} \\
 T_1 &= \sum_{\Delta(e,f,g) \in {\cal H}_3} (t_e+t_f+t_g) \label{eq:T_1} \\
 T_2 &= \sum_{\Delta(e,f,g) \in {\cal H}_3} (t_et_f+t_ft_g+t_gt_e) \label{eqn:T2} \\
 Y_1 &= S_1+ D_1-2D_2 + T_1 - 2T_2 + 3Y_3 \label{eq:modY_1}\\
 Y_2 & = D_2+T_2 -3Y_3 \label{eqn:modY_2}
 \end{align}
Note that the newly defined polynomials have the following expectations:
\begin{align*}
\mathbb{E}[S_1] &= pn_1 \\
\mathbb{E}[D_1] &= 2pn_2 \\
\mathbb{E}[D_2] &= p^2n_2 \\
\mathbb{E}[T_1] &= 3pn_3 \\
\mathbb{E}[T_2] &= 3p^2n_3 .
\end{align*}

 We observe that in the above even by change of variables $y_e=(1-t_e)$, $Y_1$ and $Y_2$ are not totally positive polynomials. This means that Proposition \ref{KIMVUCONCENTRATION} cannot be applied directly to the $Y_i's$ or $X_i$'s. The strategy we adopt is to split the $Y_1$ and $Y_2$ into many polynomials, each of which is totally positive, and then apply Proposition \ref{KIMVUCONCENTRATION} on each of them.  
  $P=\{Y_0,Y_3,S_1,D_1,D_2,T_1,T_2\}$ form the set of totally positive polynomials (proved below). Substituting the above equations into \eqref{eq:estimatorStart}-\eqref{eq:estimatorEnd}, we have the following system of equations that connect $X_i$'s and the set of totally positive polynomials $P$:
 \begin{align}
  X_0 &= Y_0 - \frac{1-p}{p}\left(S_1+ D_1+ T_1  -2D_2 -2T_2+3Y_3 \right) \nonumber \\ 
  & \qquad + \frac{(1-p)^2}{p^2}\left( D_2+T_2 -3Y_3\right) - \frac{(1-p)^3}{p^3}Y_3 \nonumber \\
  &= Y_0 - \frac{1-p}{p}(S_1 + D_1 + T_1) \nonumber \\
  & \qquad  + \frac{1-p^2}{p^2}(D_2 + T_2) - \frac{1 - p^3}{p^3}Y_3 . \label{eqn:X0} \\
  X_1 &= \frac{1}{p}\left(S_1+ D_1+ T_1  -2D_2 -2T_2+3Y_3 \right) \nonumber \\
  & \qquad - \frac{2(1-p)}{p^2}\left( D_2+T_2 -3Y_3\right) + \frac{3(1-p)^2}{p^3}Y_3 \nonumber \\
  &= \frac{1}{p} (S_1 + D_1 + T_1) - \frac{2}{p^2} (D_2 + T_2) + \frac{3}{p^3}Y_3. \label{eqn:X1}\\
  X_2 &= \frac{1}{p^2}\left( D_2+T_2 -3Y_3\right) - \frac{3(1-p)}{p^3}Y_3 \nonumber \\
  & = \frac{1}{p^2}(D_2 + T_2) - \frac{3}{p^3}Y_3. \label{eqn:X2}\\
  X_3 &= \frac{1}{p^3}Y_3 . \label{eqn:X3}
 \end{align} 

Let $\alpha_{e}$, $\beta_{e}$, and $\Delta_{e}$ be the maximum number of $H_1$'s, $H_2$'s, and $H_3$'s containing an edge $e$ in the original graph G. Let $\alpha, \beta$ and $\Delta$ be the maximum of $\alpha_e,\beta_e$, and $\Delta_e$ over all edges $e$. 
We now show concentration results for the totally positive polynomials alone.
\begin{lem}\label{lem:Y0}
 Define variables $y_e=1-t_e$. Then $Y_0$ is totally positive in $y_e$. With respect to the variables $y_e$, $\mathbb{E}_{\geq 1} \left[ Y_0 \right] \leq 3 \max \{\alpha, \beta,\Delta \}$.
\end{lem}
\begin{proof}
  We have the expectation of the following partial derivatives, up to the third order:
  \begin{align*}
    \mathbb{E} \left[ \frac{\partial Y_0}{\partial y_e} \right] &=  \alpha_e + (1-p) \beta_e + (1-p)^2 \Delta_e \nonumber \\
     \hfill &\leq 3 \max \{\alpha_e, \beta_e,\Delta_e \}. \nonumber \\
     \mathbb{E} \left[\frac {\partial Y_0}{\partial y_e y_f} \right] & \leq 1+(1-p) \leq 2 , ~
     \mathbb{E} \left[\frac{\partial Y_0}{\partial y_e y_f y_g} \right]  \leq 1. \nonumber 
  \end{align*}
  From the above equations, we have $\mathbb{E}_{\geq 1} \left[ Y_1 \right] \leq 3 \max \{\alpha, \beta, \Delta \}$ for a nonempty graph. 
  \end{proof}
 To satisfy $\mathbb{E}_{\geq1}\left[Y_0 \right] \leq \mathbb{E}[Y_0]$, it is sufficient to have
\begin{equation}\label{eqn:condY0}
n_0 \geq 3 \max \{\alpha, \beta, \Delta \} .
\end{equation}
This is because $Y_0 \geq n_0$ with probability $1$.
  \begin{lem}\label{lem:Y3}
 $Y_3$ is totally positive in $t_e$. With respect to the variables $t_e$, $\mathbb{E}_{\geq 1} \left[ Y_3 \right] \leq  \max \{1,p^2 \Delta \}$.
\end{lem}
\begin{proof}
  We have the expectation of the following partial derivatives, up to the third order:
  \begin{align*}
    \mathbb{E} \left[ \frac{\partial Y_3}{\partial t_e} \right] = p^2 \Delta_e ,~
     \mathbb{E} \left[\frac {\partial Y_3}{\partial t_e t_f} \right]  = p \leq 1 , ~
     \mathbb{E} \left[\frac{\partial Y_3}{\partial t_e t_f t_g} \right]  \leq 1. \nonumber 
  \end{align*}
  From the above equations, we have $\mathbb{E}_{\geq 1} \left[ Y_3 \right] \leq \max \{1, p^2 \Delta \}$.
\end{proof}
  $\mathbb{E}_{\geq1}\left[Y_3 \right] \leq \mathbb{E}[Y_3]$ implies
\begin{equation}\label{eqn:condY3}
 p \geq \max \{\frac{1}{\sqrt[3]{n_3}},  \Delta/n_3 \} .
\end{equation}

  \begin{lem}\label{lem:S1}
 $S_1$ is totally positive in $t_e$. With respect to the variables $t_e$, $\mathbb{E}_{\geq 1} \left[ S_1\right] \leq  \alpha$.
\end{lem}
\begin{proof}
  We have the expectation of the following partial derivatives, up to the second order:
  \begin{align*}
    \mathbb{E} \left[ \frac{\partial S_1}{\partial t_e} \right] &= \alpha_e, ~\mathbb{E} \left[\frac {\partial S_1}{\partial t_e t_f} \right] = 0 .
     \end{align*}
  From the above equations, we have $\mathbb{E}_{\geq 1} \left[ S_1 \right] \leq \alpha$ .
\end{proof}
  $\mathbb{E}_{\geq1}\left[S_1 \right] \leq \mathbb{E}[S_3]$ implies
\begin{equation}\label{eqn:condS1}
 p \geq   \alpha/n_1.
\end{equation}

 \begin{lem}\label{lem:D1}
 $D_1$ is totally positive in $t_e$. With respect to the variables $t_e$, $\mathbb{E}_{\geq 1} \left[ D_1 \right] \leq \beta$.
\end{lem}
\begin{proof}
  We have the expectation of the following partial derivatives, up to the second order:
  \begin{align*}
    \mathbb{E} \left[ \frac{\partial D_1}{\partial t_e} \right] &= \beta_e, ~\mathbb{E} \left[\frac {\partial D_1}{\partial t_e t_f} \right] = 0 .
     \end{align*}
  From the above equations, we have $\mathbb{E}_{\geq 1} \left[ D_1 \right] \leq \beta$ .
\end{proof}
  $\mathbb{E}_{\geq1}\left[D_1 \right] \leq \mathbb{E}[D_1]$ implies
\begin{equation}\label{eqn:condD1}
 p \geq   \beta/(2n_2).
\end{equation}

 \begin{lem}\label{lem:T1}
 $T_1$ is totally positive in $t_e$. With respect to the variables $t_e$, $\mathbb{E}_{\geq 1} \left[ T_1 \right] \leq \Delta$.
\end{lem}
\begin{proof}
  We have the expectation of the following partial derivatives, up to the second order:
  \begin{align*}
    \mathbb{E} \left[ \frac{\partial T_1}{\partial t_e} \right] &= \Delta_e, ~\mathbb{E} \left[\frac {\partial T_1}{\partial t_e t_f} \right] =0 .
     \end{align*}
  From the above equations, we have $\mathbb{E}_{\geq 1} \left[ T_1 \right] \leq \Delta$ .
\end{proof}
  $\mathbb{E}_{\geq1}\left[T_1 \right] \leq \mathbb{E}[T_1]$ implies
\begin{equation}\label{eqn:condT1}
 p \geq   \Delta/(3n_3).
\end{equation}

\begin{lem}\label{lem:D2}
 $D_2$ is totally positive in $t_e$. With respect to the variables $t_e$, $\mathbb{E}_{\geq 1} \left[ D_2 \right] \leq \max \{p\beta,1\}$.
\end{lem}
\begin{proof}
  We have the expectation of the following partial derivatives, up to the second order:
  \begin{align*}
    \mathbb{E} \left[ \frac{\partial D_2}{\partial t_e} \right] &= p\beta_e, ~\mathbb{E} \left[\frac {\partial D_2}{\partial t_e t_f} \right] \leq 1 .
     \end{align*}
  From the above equations, we have $\mathbb{E}_{\geq 1} \left[ D_2 \right] \leq \max \{p\beta,1\}$.
\end{proof}
  $\mathbb{E}_{\geq1}\left[D_2 \right] \leq \mathbb{E}[D_2]$ implies
\begin{equation}\label{eqn:condD2}
 p \geq  \max \{ \beta/n_2, \frac{1}{\sqrt{n_2}}  \}.
\end{equation}

\begin{lem}\label{lem:T2}
  $T_2$ is totally positive in $t_e$. With respect to the variables $t_e$, $\mathbb{E}_{\geq 1} \left[ T_2 \right] \leq \max \{2p\Delta,1\}$.
\end{lem}
\begin{proof}
  We have the expectation of the following partial derivatives, up to the second order:
  \begin{align*}
    \mathbb{E} \left[ \frac{\partial T_2}{\partial t_e} \right] &= 2p\Delta_e, ~\mathbb{E} \left[\frac {\partial T_2}{\partial t_e t_f} \right] \leq 1 .
     \end{align*}
  From the above equations, we have $\mathbb{E}_{\geq 1} \left[ T_2 \right] \leq \max \{2p\Delta,1\}$.
\end{proof}
  $\mathbb{E}_{\geq1}\left[T_2 \right] \leq \mathbb{E}[T_2]$ implies
\begin{equation}\label{eqn:condT2}
 p \geq  \max \{ 2\Delta/(3n_3), \frac{1}{\sqrt{3n_3}}  \}.
\end{equation}

Now merging all the conditions \eqref{eqn:condY0}-\eqref{eqn:condT2}, we get
\begin{align} \label{eqn:megacond1}
\begin{split}
  n_0 &\geq 3 \max\{\alpha,\beta,\Delta\}\\
   p &\geq \max \{ \frac{1}{\sqrt[3]{n_3}},\frac{1}{\sqrt{n_2}} ,\frac{\Delta}{n_3}, \frac{\beta}{n_2} ,\frac{\alpha}{n_1}\}.
 \end{split}
\end{align}

Applying Proposition \ref{KIMVUCONCENTRATION} to all the totally positive polynomials, along with \eqref{eqn:megacond1}, we get
\begin{align}
\begin{split}
\mathbb{P}\left( |Y_0 - \mathbb{E}[Y_0] | > a_3 \sqrt{\mathbb{E}[Y_0] \mathbb{E}_{\geq 1}[Y_0]}  \lambda_1^3 \right)   \\
= \mathcal{O}\left(\exp \left(-\lambda_1 + (2) \log m \right) \right)    \\
\mathbb{P}\left( |Y_3- \mathbb{E}[Y_3] | > a_3 \sqrt{\mathbb{E}[Y_3] \mathbb{E}_{\geq 1}[Y_3]}  \lambda_2^3 \right)  \\
= \mathcal{O}\left(\exp \left(-\lambda_2 + (2) \log m \right) \right)  \\
\mathbb{P}\left( |S_1 - \mathbb{E}[S_1] | > a_1 \sqrt{\mathbb{E}[S_1] \mathbb{E}_{\geq 1}[S_1]}  \lambda_3 \right)  \\
= \mathcal{O}\left(\exp \left(-\lambda_3  \right) \right)  \\
\mathbb{P}\left( |D_1 - \mathbb{E}[D_1] | > a_1 \sqrt{\mathbb{E}[D_1] \mathbb{E}_{\geq 1}[D_1]}  \lambda_4 \right)  \\
= \mathcal{O}\left(\exp \left(-\lambda_4  \right) \right)  \\
\mathbb{P}\left( |T_1 - \mathbb{E}[T_1] | > a_1 \sqrt{\mathbb{E}[T_1] \mathbb{E}_{\geq 1}[T_1]}  \lambda_5  \right)  \\
= \mathcal{O}\left(\exp \left(-\lambda_5   \right) \right)  \\
\mathbb{P}\left( |D_2 - \mathbb{E}[D_2] | > a_2 \sqrt{\mathbb{E}[D_2] \mathbb{E}_{\geq 1}[D_2]}  \lambda_6^2 \right)  \\
= \mathcal{O}\left(\exp \left(-\lambda_6 +\log m \right) \right)  \\
\mathbb{P}\left( |T_2 - \mathbb{E}[T_2] | > a_2 \sqrt{\mathbb{E}[T_2] \mathbb{E}_{\geq 1}[T_2]}  \lambda_7^2 \right)  \\
= \mathcal{O}\left(\exp \left(-\lambda_7 + \log m \right) \right) . \label{eqn:equal}
\end{split}
\end{align}

Choose an $\epsilon>0$. We force the following conditions:
 \begin{align}
 \begin{split}
 a_3 \sqrt{\mathbb{E}[Y_0] \mathbb{E}_{\geq 1}[Y_0]}  \lambda_1^3 &= \epsilon \mathbb{E}[Y_0] \\
 a_3 \sqrt{\mathbb{E}[Y_3] \mathbb{E}_{\geq 1}[Y_3]}  \lambda_2^3 &= \epsilon \mathbb{E}[Y_3]  \\
 a_1 \sqrt{\mathbb{E}[S_1] \mathbb{E}_{\geq 1}[S_1]}  \lambda_3 &= \epsilon \mathbb{E}[S_1]  \\
 a_1 \sqrt{\mathbb{E}[D_1] \mathbb{E}_{\geq 1}[D_1]}  \lambda_4  &= \epsilon \mathbb{E}[D_1]  \\
 a_1 \sqrt{\mathbb{E}[T_1] \mathbb{E}_{\geq 1}[T_1]}  \lambda_5  &= \epsilon \mathbb{E}[T_1] \\
 a_2 \sqrt{\mathbb{E}[D_2] \mathbb{E}_{\geq 1}[D_2]}  \lambda_6^2 &= \epsilon \mathbb{E}[D_2]  \\
 a_2 \sqrt{\mathbb{E}[T_2] \mathbb{E}_{\geq 1}[T_2]}  \lambda_7^2 &= \epsilon \mathbb{E}[T_2] .
\end{split}
 \end{align}
 
 Let $\gamma>0$. For the right hand side of every equation in \eqref{eqn:equal} to be $\mathcal{O}(\exp(- \gamma \log m))$, assuming all the bounds in Lemmas \ref{lem:Y0}-\ref{lem:T2}, it is sufficient to have
\begin{align}
\begin{split}
\frac{n_0}{3 \max\{\alpha,\beta,\Delta\}} &\geq \frac{a_3^2 \log^6 \left(m^{2+\gamma} \right)}{\epsilon^2}  \\
\frac{p} { \max \{\frac{1}{\sqrt[3]{n_3}},  \Delta/n_3 \} } &\geq \frac{a_3^2 \log^6 \left(m^{2+\gamma} \right)}{\epsilon^2}  \\
\frac{p} {\max\{\frac{\alpha}{n_1},\frac{\beta}{2n_2},\frac{\Delta}{3n_3}\} } &\geq \frac{a_1^2 \log^2 \left(m^{\gamma} \right)}{\epsilon^2}  \\
\frac{p}{\max \{\frac{\beta}{n_2},\frac{2\Delta}{3n_3},\frac{1}{\sqrt{n_2}},\frac{1}{\sqrt{n_3}} \}} & \geq  \frac{a_2^2 \log^4 \left(m^{1+\gamma} \right)}{\epsilon^2} . \label{eq:condprefinal}
\end{split}
\end{align}

We can see that the conditions in \eqref{eq:condprefinal} imply the conditions in \eqref{eqn:megacond1}. These can be simplified to remove some redundancy as follows:
\begin{align}
\begin{split}
\label{eqn:condfinal}
\frac{n_0}{3 \max\{\alpha,\beta,\Delta\}} &\geq \frac{a_3^2 \log^6 \left(m^{2+\gamma} \right)}{\epsilon^2}  \\
\frac{p} { \max \{\frac{1}{\sqrt[3]{n_3}},  \Delta/n_3 \} } &\geq \frac{a_3^2 \log^6 \left(m^{2+\gamma} \right)}{\epsilon^2}  \\
\frac{p} {\alpha/n_1 } &\geq \frac{a_1^2 \log^2 \left(m^{\gamma} \right)}{\epsilon^2}  \\
\frac{p}{\max \{\frac{\beta}{n_2},\frac{1}{\sqrt{n_2}} \}} & \geq  \frac{a_2^2 \log^4 \left(m^{1+\gamma} \right)}{\epsilon^2} .
\end{split}
\end{align}
This is due to the fact that $a_3 \geq a_2 \geq a_1$ and $m^3 \geq m^2$. Therefore, subject to \eqref{eqn:condfinal}, all totally positive polynomials $Y_0,Y_3,D_1,D_2,S_1,T_1,T_2$ concentrate within a multiplicative factor of $(1\pm \epsilon)$ with probability at least $1-\mathcal{O}\left(\frac{1}{m^{\gamma}} \right)$.

Under the above concentration result, let the deviations of $X_i$'s be denoted by $\delta X_i$. Now we calculate the deviation of $X_0$ using \eqref{eqn:X0}.
\begin{align}
 \delta X_0 & \leq \epsilon \mathbb{E}[Y_0] + \epsilon \frac{1-p}{p}\left( \lvert \mathbb{E}[S_1] \rvert+ \lvert \mathbb{E}[D_1] \rvert+ \lvert \mathbb{E}[T_1]\rvert \right) \nonumber\\
  &  \qquad + \epsilon \frac{1-p^2}{p^2} \left(  \lvert\mathbb{E} [D_2] \rvert +  \lvert \mathbb{E}[T_2] \rvert \right)   - \epsilon \frac{1-p^3}{p^3} \lvert \mathbb{E}[Y_3] \rvert \nonumber\\
	\hfill & \leq \epsilon (n_0 + n_1 + 3n_2 + 7n_3 ) \nonumber \\ 
	\hfill & \leq 7 \epsilon (n_0+n_1+n_2+n_3). \nonumber
\end{align}

Similarly for other $X_i$'s, we get
 \begin{align*}
   \delta X_1 &\leq 12 \epsilon \left(n_1+n_2+n_3 \right) \\
   \delta X_2 &\leq 6\epsilon \left(n_2+n_3 \right) \\
    \delta X_3 &\leq \epsilon n_3 .
 \end{align*}

Therefore, sampling every edge independently with probability $p$ satisfying all conditions in (\ref{eqn:condfinal}), all $X_i's$ concentrate within an additive gap of $(1 \pm 12 \epsilon) {\lvert V \rvert \choose 3}$ with probability at least $1- \frac{1}{m^{\gamma}}$. The constants in this proof can be tightened by a more accurate analysis.

\end{document}